\documentclass[10pt,a4paper]{article}

\usepackage{amstext,amsgen,latexsym}
\usepackage{amstext,amssymb,amsfonts,latexsym}
\usepackage{theorem}
\usepackage{pifont}

 \newcommand{\bs}{\bigskip}
 \newcommand{\ms}{\medskip}
 \newcommand{\n}{\noindent}
 \newcommand{\s}{\smallskip}
 \newcommand{\hs}[1]{\hspace*{ #1 mm}}
 \newcommand{\vs}[1]{\vspace*{ #1 mm}}

 \newcommand{\real}{\mathbb{R}}

 \newcommand{\nat}{\mathbb{N}}
 
 \newcommand{\integer}{\mathbb{Z}}

 \newcommand{\prob}{{\mathrm{Prob}}}

 \newcommand{\co}{\mathrm{co}\mbox{-}}

 \newcommand{\ie}{\textrm{i.e.},\hspace*{2mm}}
 \newcommand{\eg}{\textrm{e.g.},\hspace*{2mm}}
 
 \newcommand{\etalc}{\textrm{et al.}}

 \newcommand{\AAA}{{\cal A}}
 
 \newcommand{\CC}{{\cal C}}
 
 \newcommand{\DD}{{\cal D}}

 \newcommand{\p}{\mathrm{P}}
 \newcommand{\np}{\mathrm{NP}}

 \newcommand{\oneflin}{1\mbox{-}\mathrm{FLIN}}

 \newcommand{\onedlin}{1\mbox{-}\mathrm{DLIN}}

 \newcommand{\reg}{\mathrm{REG}}
 \newcommand{\cfl}{\mathrm{CFL}}
 \newcommand{\dcfl}{\mathrm{DCFL}}

\newcommand{\comb}[2]{\left({\small \begin{array}{c} #1 \\%
      #2 \end{array} }\right)}
      
 \newcommand{\tinycomb}[2]{\left({\tiny \begin{array}{c} #1 \\%
      #2 \end{array} }\right)}

 \newcommand{\IFF}{\Longleftrightarrow}

 \def\bbox{\vrule height6pt width6pt depth1pt}

\theoremstyle{plain}
\theoremheaderfont{\bfseries}
 \newtheorem{theorem}{Theorem}[section]
 \newtheorem{lemma}[theorem]{Lemma}
 \newtheorem{proposition}[theorem]{Proposition}

 {\theorembodyfont{\rmfamily}
  }
 {\theorembodyfont{\rmfamily} }
 {\theorembodyfont{\rmfamily} }

 \newtheorem{claim}{Claim}

 \newenvironment{proof}{\par \noindent
            {\bf Proof. \hs{2}}}{\hfill$\Box$ \vspace*{3mm}}

 \newenvironment{proofof}[1]{\vspace*{5mm} \par \noindent
         {\bf Proof of #1.\hs{2}}}{\hfill$\Box$ \vspace*{3mm}}

 \newcommand{\ceilings}[1]{\lceil #1 \rceil}
 \newcommand{\floors}[1]{\lfloor #1 \rfloor}

\newif\ifnotesw\noteswtrue
   {\ifnotesw\marginpar[\hfill\(\top\)]{\(\top\)}\fi}%
      {\ifnotesw\marginpar[\hfill\(\bot\)]{\(\bot\)}\fi}
      
\newcommand{\mnote}[1]%
   {\ifnotesw\marginpar%
	  [{\scriptsize\begin{minipage}[t]{\marginparwidth}
	  \raggedleft#1%
		  \end{minipage}}]%
	  {\scriptsize\begin{minipage}[t]{\marginparwidth}
	  \raggedright#1%
		  \end{minipage}}%
    \fi}

\newcommand{\ignore}[1]{}

\newcommand{\track}[2]{[{\tiny \begin{array}{c} #1 \\%
      #2 \end{array} }]}
    
\newcommand{\fpartial}{\mathrm{partial}}

\begin{document}
\begin{center}
{\Large {\bf Immunity and Pseudorandomness of \s\\
Context-Free Languages}} \bs\\
{\sc Tomoyuki Yamakami}\footnote{Affiliation at the time of the first version: School of Computer Science and Engineering, University of Aizu, 90 Kami-Iawase, Tsuruga, Ikki-machi, Aizu-Wakamatsu, Fukushima 965-8580, Japan.} \bs\\
\end{center}

\pagestyle{plain}

\begin{quote}
{\bf Abstract.}\hs{1}
We discuss the computational complexity of context-free languages,  concentrating on two well-known structural properties---immunity and pseudorandomness. An infinite language is REG-immune (resp., CFL-immune) if it contains no infinite subset that is a regular (resp., context-free) language.
We prove that (i) there is a context-free REG-immune language outside REG/$n$ and  (ii) there is a REG-bi-immune language that can be computed deterministically using logarithmic space.
We also show that (iii) there is a CFL-simple set, where a CFL-simple language is an infinite context-free language whose complement is CFL-immune. Similar to the REG-immunity, a 
REG-primeimmune language has no polynomially dense subsets that are also regular. 
We further prove that (iv) there is a context-free language that is 
REG/$n$-bi-primeimmune. 
Concerning  pseudorandomness of context-free languages, we show that (v) CFL contains   REG/$n$-pseudorandom languages. Finally, we prove that  (vi) against REG/$n$, there exists an almost 1-1 pseudorandom generator computable in nondeterministic pushdown automata equipped with a write-only output tape and (vii) against REG, there is no almost 1-1 weakly pseudorandom generator computable  deterministically in linear time by a single-tape Turing machine.  

\ms


\end{quote}

\section{Motivations and a Quick Overview}

The notion of {\em context-free languages} is one of the most fundamental concepts in formal language theory. Besides its theoretical interest, the context-freeness  has drawn, since the 1960s, practical applications in key fields of computer science, including programing languages, compiler implementation, and markup languages, mainly attributed to unique traits of context-free grammars or phrase-structure grammars. Some of the traits can be highlighted by, for instance,  pumping and swapping lemmas \cite{BPS61,Yam08}, normal form theorems \cite{Cho59,Gre65}, and undecidability theorems  \cite{BPS61,GR63}, all of which reveal certain hidden 
substructures of the context-free languages.  
The literature over half a century has successfully explored numerous basic properties (inclusive of operational closure, normal forms, and minimization) of the family $\cfl$ of all context-free languages. 
We wish to continue promoting our understandings of  $\cfl$ further. 
This family $\cfl$ contains a number of non-regular languages, such as $L_{eq}=\{0^n1^n\mid n\geq0\}$ and $Equal =\{w\in\{0,1\}^*\mid \#_{0}(w)=\#_{1}(w)\}$, where 
$\#_{b}(w)$ denotes the number of $b$'s in $w$. 
An effective use of a pumping lemma, for example, easily separates them from the family $\reg$ of regular languages (see, \eg \cite{HMU01} for their proofs). Nonetheless, these two context-free languages look quite different in nature and in complexity. 
How different is one language from another? How can we exactly describe 
a ``complex'' nature of those languages? 
These questions that arise naturally motivate us 
to search for a suitable ``complexity measure.''   
Since  time-complexity is not a suitable complexity measure for the context-free languages, another simple way to scale their complexity
is to show ``structural'' differences among those languages.
  
Up to now, numerous structural properties have been proposed for  
polynomial-time complexity classes, such as $\p$ (deterministic polynomial-time class) and $\np$ (nondeterministic polynomial-time class), and have been studied to understand their behaviors and also  characteristics. Many of those properties have arisen naturally in a context of answering long-unsettled questions, including 
the famous $\p=?\np$ question (see, \eg \cite{BDG88} for those properties). 
To measure the complexity of each context-free language, we intend to target two well-known structural properties---{\em immunity} and {\em pseudorandomness}---which have been studied since the 1940s in computational complexity theory and computational cryptography. These two properties are known to be closely related. In this paper, we shall spotlight them within a framework of formal language theory. This framework makes it possible to prove many properties (such as the existence of $\cfl$-immune languages), without any unproven assumption or 
any relativization, by taking approaches that  are quite different from standard ones in a setting of polynomial-time bounded computation.

In the first part of this paper (Sections \ref{sec:immunity-notion}--\ref{sec:p-dense-immune}), our special attention goes to languages that have only ``computationally-hard'' non-trivial subsets. 
Those languages, known as {\em immune} languages and {\em simple} languages, naturally possess high complexity. 
Formally, given a fixed family $\CC$ of languages, 
an infinite language is {\em $\CC$-immune} if it has no infinite 
subset in $\CC$, and a {\em $\CC$-simple} language is an infinite language in $\CC$ whose complement is $\CC$-immune. Significantly, the $\CC$-immunity satisfies a {\em self-exclusion property}: $\CC$ cannot be $\CC$-immune. 
Notice that  the notion of simplicity has played a key role in the theory of $\np$-completeness (see, \eg \cite{BDG88}).  In addition, a language 
is called {\em $\CC$-bi-immune} if its complement and itself are both $\CC$-immune. 

These notions of immunity and simplicity date back to the 1940s, in which they were first conceived by Post \cite{Pos44} for recursively enumerable languages (see, \eg \cite{Rog67}). Their resource-bounded analogues were discussed later in the 1970s by Flajolet and Steyaert \cite{FS74}. During the 1980s, Ko and Moore \cite{KM81} intensively studied 
such limited immunity, whereas Homer and Maass \cite{HM83} explored resource-bounded simplicity. The bi-immunity notion was introduced in mid-1980s by Balc{\'a}zar and Sch{\"o}ning \cite{BS85}.  
Since then, numerous variants of immunity and simplicity (for instance, strong immunity, almost immunity, balanced immunity, and hyperimmunity) have been proposed and studied extensively (see, \eg \cite{BDG88,YS05} for references therein). 

Despite the past efforts in a setting of polynomial-time  bounded computation, the immunity notion has eluded from our full understandings; for instance, it has been open whether there exists a $\p$-immune set in $\np$ or even an $\np$-simple set since the existence of such  a set immediately yields a class separation between $\np$ and $\co\np$.  
Only in relativized worlds, we can prove directly 
the existence of those immune and simple sets (see, \eg \cite{Bal84,BS85,HM83,Lis99,SB84}). 
While there is a large volume of work on the immunity of polynomial-time complexity classes, there has been little study done on the immunity of the  context-free languages since the work of Flajolet and Steyaert. 
We expect that an 
analysis of $\reg$-immunity inside $\cfl$ would bring into new light a structural difference among various context-free languages. 
For instance, the aforementioned context-free language $L_{eq}$ is 
 $\mathrm{REG}$-immune \cite{FS74}, whereas its accompanied  language $Equal$ is not $\mathrm{REG}$-immune. 
Moreover, we can prove many structural properties with no extra unproven assumptions or even no relativization. For instance, unlike the case of $\np$-simplicity, a direct argument demonstrates that $\mathrm{CFL}$-simple languages actually exist. 
As those examples suggest, {\em context-freeness} provides tremendous advantages of proving immunity as well as 
simplicity over polynomial-time complexity classes. 

Nonetheless, all questions concerning the $\reg$-immunity in $\cfl$ have not settled in this paper. One of those unsettled questions is related to 
$\reg$-bi-immunity. 
It is unclear that $\mathrm{REG}$-bi-immune languages actually exist inside $\cfl$. At our best, we can prove that the language class $\mathrm{L}$ (deterministic logarithmic-space class) contains $\reg$-bi-immune languages.  
Another unsolved question concerns a density issue of immune languages. Notice that all known $\reg$-immune languages $L$ in $\cfl$ have exponentially-small density rate $|L\cap\Sigma^n|/|\Sigma^n|$. The $\reg$-immune language  
$L_{eq}$, for instance,  has density rate $|L_{eq}\cap\{0,1\}^n|/2^n \leq  1/2^{n}$ 
for each even length $n$; in contrast, $Equal$, which is not even $\reg$-immune,  
has its density rate $|Equal\cap\{0,1\}^n|/2^n \geq 1/n$ 
for any sufficiently large even number $n$.  
Naturally, we can ask whether there exists any context-free  
$\reg$-immune language whose density $|L\cap\Sigma^n|$ is lower-bounded by a ``polynomial'' fraction, \ie $1/p(n)$ for a certain non-zero polynomial $p$. Such a density condition is referred to as {\em polynomially dense} or {\em p-dense}. 
In this paper, as the first step toward the above open question, 
we can show the existence of 
a p-dense $\reg$-immune language in $\mathrm{L}$. 
The difficulty of proving those structural properties 
of $\cfl$ might indicate a limitation of the expressing power of context-freeness as languages. 

Recall that $\CC$-immunity requires the non-existence of 
an infinite subset in $\CC$. Is there any language that lacks only p-dense subsets (instead of all infinite subsets) in $\CC$? Such a natural question gives rise to a variant of $\CC$-immunity, referred to as {\em $\CC$-primeimmunity}. Now, we turn our attention to this 
new notion inside $\cfl$. 
With a slightly adroit argument, we can prove that an 
``extended'' language of $Equal$, $Equal_{*} = \{aw\mid a\in\{\lambda,0,1\},w\in Equal\}$, is $\reg/n$-primeimmune, where $\reg/n$ is obtained from $\reg$ by supplementing appropriate ``advice'' of size $n$ \cite{TYL04,Yam08}. In stark contrast to the  $\reg$-bi-immunity, 
we can show that $\reg$-bi-primeimmune languages (even $\reg/n$-bi-primeimmune languages) exist inside $\cfl$.


The second part of this paper (Sections \ref{sec:pseudorandom}--\ref{sec:generator}) is exclusively devoted to a property of computational randomness, or {\em pseudorandomness}. 
An early computational approach to ``randomness'' began in the 1940s. Church's \cite{Chu40} random 0-1 sequences, for instance, 
demand that every infinite subsequence should contain 
{\em asymptotically} the same number of 0s and 1s. This line of study on computational randomness, also known as {\em stochasticity}, concerns asymptotic behaviors of random sequences.   It has been known a close 
connection between stochasticity and bi-immunity. 

To suit our study of the context-free languages, however, we rather 
look into  ``non-asymptotic'' behaviors of randomness 
inside languages. This paper  
discusses the following type of ``random'' languages. We say that a language $L$ is {\em  $\CC$-pseudorandom} if, for every language $A$ in $\CC$, the characteristic function $\chi_{A}$ agrees with $\chi_{L}$ on ``nearly'' $50\%$ of strings of each length, where ``nearly'' means ``with a negligible margin of error.'' Our notion can be seen as a variant of Wilber's \cite{Wil83} randomness, which dictates an asymptotic behavior of  $\chi_{L}$ and $\chi_{A}$.  

Similar in the case of primeimmunity, p-denseness requires our special attention. Targeting p-dense languages, we introduce another ``randomness'' notion, called {\em weak $\CC$-pseudorandomness}, as a non-asymptotic variant of M{\"u}ller's \cite{Mul93} balanced immunity,  Loveland's \cite{Lov66} unbiasedness, and weak-stochasticity of Ambos-Spies 
\etalc~\cite{AMWZ96}.
Loosely speaking, a language $L$ is weak $\CC$-pseudorandom if the density rate $|L\cap A\cap\Sigma^n|/|A\cap\Sigma^n|$ is close to $1/2$ for every p-dense language $A$ in $\CC$. 

A typical example of $\reg/n$-pseudorandom language is the set $IP_{*}$ whose strings are of the form $auv$ with $a\in\{\lambda,0,1\}$ and $|u|=|v|$ such that the binary inner product between $u^R$ and $v$ is odd. 
A close connection between pseudorandomness and primeimmunity  
draws a conclusion that $IP_{*}$ is also $\reg/n$-bi-primeimmune. By  clear contrast, the aforementioned language $Equal_{*}$, for instance, can separate the notion of $\reg/n$-primeimmunity 
from the notion of weak $\reg/n$-pseudorandomness. 

In the early 1980s, Blum and Micali \cite{BM84} studied {\em pseudorandom generators}, which produce unpredictable sequences. Our formulation of pseudorandom generators, attributed to Yao \cite{Yao82}, uses indistinguishability from uniform sequences. Loosely speaking, a pseudorandom generator is a function producing a string that looks random for any target adversary (in this case, the generator is said to {\em fool} it). In our language setting, we call a function mapping $\Sigma^*$ to $\Sigma^*$ with stretch factor $s(n)$ (that is, $|f(x)|=s(|x|)$) a {\em pseudorandom generator} against a language family $\CC$ if 
$G$ fools every language in $\CC$. Our pseudorandom generator actually 
tries to fool languages in a sense that, over string inputs of each length $n$, the outcome distribution of the generator is indistinguishable 
from the strings of length $s(n)$; namely, 
the function $\ell(n)=|\prob_{x}[\chi_{A}(x)=1] - \prob_{y}[\chi_{A}(y)=1]|$ has negligibly small values, where $x$ and $y$ are chosen uniformly at random from $\Sigma^n$ and $\Sigma^{s(n)}$, respectively. We can 
prove that, against the language family $\reg/n$, there exists an almost 1-1 pseudorandom generator computable by a nondeterministic pushdown automaton equipped with an output tape.   As a limitation of the power of 
generators, we can show that, even against $\reg$, there is no almost 1-1 pseudorandom generator computable by a one-tape one-head linear-time deterministic Turing machine.

\section{Foundations}\label{sec:notation}

The {\em natural numbers} are nonnegative integers and we write $\nat$ to denote the set of all natural numbers. We set $\nat^{+}=\nat-\{0\}$ for convenience. For any two integers $m,n$ with $m\leq n$, the notation $[m,n]_{\integer}$ stands for the integer interval $\{ m,m+1,m+2,\ldots,n\}$. The {\em symmetric difference} between two sets $A$ and $B$, denoted $A \triangle B$, is the set $(A-B)\cup (B-A)$. 
In this paper, all logarithms are assumed to have base two unless otherwise stated.
Let $\log^{(1)}{n}=\log{n}$ and $\log^{(i+1)}{n} = \log(\log^{(i)}{n})$ for each number $i\in\nat^{+}$. 
A function $\mu$ from $\nat$ to $\real^{\geq0}$ (all nonnegative reals) is called {\em noticeable} if there exists a non-zero polynomial $p$ such that $\mu(n)\geq 1/p(n)$ for all but finitely many numbers $n$ in $\nat$.  By contrast, $\mu$ is called {\em negligible} if we have  $\mu(n)\leq 1/p(n)$  for any non-zero polynomial $p$ and for all sufficiently large numbers $n\in\nat$. 

Our {\em alphabet}, often denoted $\Sigma$, is always a nonempty finite set. A {\em string} is a series of symbols taken from $\Sigma$, and the {\em length} of a string $x$ is the number of symbols in $x$ and is denoted $|x|$. 
The {\em empty string} is always denoted $\lambda$ and,  
for two strings $x$ and $y$, $xy$ denotes the {\em concatenation} of $x$ and $y$. In particular, $\lambda x$ coincides with $x$. The notation $\Sigma^n$ denotes the set of all strings of length $n$. 
For any string $x$ of length $n$ and for any index $i\in[0,n]_{\integer}$,  $pref_{i}(x)$ is the substring of $x$, made up with the first $i$ symbols of $x$. In particular, we have $pref_{0}(x)=\lambda$.   
For each string $w\in\Sigma^*$ and any symbol $a\in\Sigma$, 
the number of $a$'s appearing in $w$ is represented by $\#_{a}(w)$.  
A {\em language} over an alphabet $\Sigma$ is a subset of $\Sigma^*$, and the {\em characteristic function} $\chi_{A}$ of $A$ is defined as $\chi_{A}(x)=1$ if $x\in A$ and $\chi_{A}(x)=0$ otherwise for every string $x\in\Sigma^*$.

For any language $L$ over $\Sigma$, the {\em complement} 
of $L$ (\ie  $\Sigma^* - L$) is often denoted $\overline{L}$ whenever $\Sigma$ is clear from the context. Furthermore, the {\em complement} of a family $\CC$ of languages is the collection of all languages whose complements are in $\CC$. We use the conventional notation $\co\CC$ to denote the complement of $\CC$. For simplicity, 
the notation $dense(L)(n)$ expresses the cardinality of the set $L\cap\Sigma^n$; that is, $dense(L)(n)= |L\cap\Sigma^{n}|$. A language $L$ over $\Sigma$ 
is called {\em (polynomially) sparse} if $dense(L)(n)$ is upper-bounded by a certain fixed polynomial in $n$. 

Since this paper mainly discusses {\em regular languages} and {\em context-free languages}, we assume the reader's basic knowledge on 
fundamental mechanisms of one-tape one-head one-way finite automata, possibly equipped with pushdown (or first-in last-out) stacks. See, \eg \cite{HU79,HMU01} for the formal definitions of these finite automata. 
Generally speaking, for each finite automaton $M$, the notation $L(M)$ represents the set of all strings ``accepted'' by $M$ under appropriate accepting criteria. Notice that such criteria may significantly differ if we choose different machine types. Conventionally, we say that $M$ {\em recognizes} a language $L$ if $L=L(M)$. 
Languages recognized by {\em deterministic finite automata} (or dfa's) and {\em nondeterministic pushdown automata} (or npda's) are respectively called {\em regular languages} and  {\em context-free languages}. 
For ease of notation, 
we denote by $\mathrm{REG}$ the family of regular languages and by $\mathrm{CFL}$ the family of context-free languages. 
In addition, {\em deterministic pushdown automata} (or dpda's) recognize only {\em deterministic context-free languages}, and $\mathrm{DCFL}$ denotes the family of all deterministic context-free languages. 

It is known that the language family $\cfl$ is not closed under conjunction (see, \eg \cite{HMU01} for the proof). This fact inspires us to introduce a restricted conjunctive closure of $\cfl$. For any positive integer $k$, the {\em $k$ conjunctive closure} of $\cfl$, denoted $\cfl(k)$, is the collection of all languages $L$ such that there are $k$ languages $L_1,L_2,\ldots,L_k$ in $\cfl$ for which $L = L_1\cap L_2\cap\cdots \cap L_k$. By its definition, $\cfl(1)$ coincides with $\cfl$ itself.

To explain the notion of {\em advice}, we first adapt a ``track'' notation $\track{x}{y}$ from \cite{TYL04}. For any pair of symbols $\sigma\in\Sigma_1$ and $\tau\in\Sigma_2$, the notation $\track{\sigma}{\tau}$ denotes a new symbol made from $\sigma$ and $\tau$. For two strings $x=x_1x_2\cdots x_n$ and $y=y_1y_2\cdots y_n$ of the same length $n$, the notation $\track{x}{y}$ is shorthand for the string $\track{x_1}{y_1}\track{x_2}{y_2} \cdots\track{x_n}{y_n}$.  
An {\em advice function} is a map from $\nat$ to $\Gamma^*$, where $\Gamma$ is an appropriate alphabet. For any family $\CC$ of languages, the advised class $\CC/n$ denotes the collection of languages $L$ over an alphabet $\Sigma$ for which there exist another alphabet $\Gamma$, an advice function $h:\nat\rightarrow\Gamma^*$, and a language $A\in\CC$ such that, for every string $x\in\Sigma^*$, (i) $|h(|x|)|=|x|$ (\ie length preserving) and (ii)  $x\in L$ iff $\track{x}{h(|x|)}\in A$ \cite{TYL04,Yam08}.

As an additional computation model, we introduce the notion of 
one-tape one-head off-line Turing machines  
whose tape heads move in all directions. Such machines are succinctly called {\em 1TMs}. 
All tape cells of an infinite input/work tape are indexed with integers and an input string of length $n$ is given in the cells indexed between $1$ and $n$ surrounded by two designated {\em endmarkers}. We take a notation $\mathrm{1}\mbox{-}\mathrm{DTIME}(t(n))$ from \cite{TYL04} to denote the collection of all languages that are recognized within time $t(n)$ 
by those 1TMs. As a special case, we write $\onedlin$ for $\mathrm{1\mbox{-}DTIME}(O(n))$. It is well-known that $\reg= 
\onedlin =  \mathrm{1\mbox{-}DTIME}(o(n\log{n}))$ \cite{Hen65,Kob85}.

To handle (multi-valued partial) functions, we further consider  Turing machines that produce (possibly) many output strings at once. Conventionally, whenever a single-tape machine halts along the tape that contains only a block of non-blank symbols beginning at the left endmarker  and surrounded only by blanks, we treat the string given in this block as an {\em outcome} of the machine. A (partial) function $f$ from $\Sigma^*$ to $\Gamma^*$, where $\Sigma$ and $\Gamma$ are two alphabets, is called {\em length preserving} if $|f(x)|=|x|$ for any string $x$ in the domain of $f$.

Let us introduce several function classes, which are natural extensions of the language families $\reg$ and $\cfl$. The function class $\oneflin$ is the set of all {\em single-valued total} functions computable in time $O(n)$ by deterministic 1TMs. Similarly, the notation $\oneflin(\fpartial)$ expresses the set of all {\em single-valued partial} functions $f$ such that there exists a linear-time deterministic 1TM $M$ that starts with input $x$ and halts with output $f(x)$ by entering an accepting state whenever $f(x)$ is defined; $M$ always 
enters a rejecting state when $f(x)$ is not defined. 

We expand single-valued functions to multi-valued functions, which 
produce sets of values. We define $\mathrm{1\mbox{-}NLINMV}$ as 
the class of all  {\em multi-valued partial} functions $f$ for which there exists a nondeterministic 1TM $M$, provided that all computation (both accepting and rejecting) paths terminates with certain output values in time $O(n)$, together with the condition that $f(x)$ consists of all output values produced along accepting paths. Notice that, when $f(x)=\emptyset$, there should be no accepting path. See \cite{TYL04} for their basic properties.

The original npda model was introduced to recognize ``languages.'' 
Let us expand this model to compute (partial) functions. For this purpose, we equip an npda  with an additional {\em output tape} and its associated tape head. Now, our npda has two tapes: a {\em read-only} input tape and a {\em write-only} output tape. 
This new npda acts as a standard npda with a single stack except for moves of an output-tape head. In the write-only output tape, its tape head always moves to the right whenever it writes a non-blank symbol in its tape cell. 
Here, we allow the tape head to stay still on a blank symbol 
as long as it does not write any non-blank symbol. Since the head moves only to a new blank cell, it cannot read any meaningful symbol that have already written in the output tape. Along each computation path, we define an {\em output} of the npda as follows.  When the npda enters an accepting state, we treat the string produced on the output tape as an output of the machine. On the contrary, when the machine enters a rejecting state, we assume that the machine produces no output along this path although 
there may be non-blank symbols left on the output tape. 
Hence, the machine can produce at least one output value or no output value at all. 
Therefore, such an npda in general computes a multi-valued partial function.
Let $\mathrm{CFLMV}$ denote the collection of all multi-valued partial functions that can be produced by those npda's. Moreover, 
$\mathrm{CFLSV}$ consists of all {\em single-valued partial} functions in $\mathrm{CFLMV}$. When the functions $f$ are limited to be total (\ie $f(x)$ is always defined), we use the notation $\mathrm{CFLSV_t}$. Note that, for every language $L$, $L\in\cfl$ iff $\chi_{L}\in\mathrm{CFLSV_t}$. 

\section{Resource-Bounded Immunity and Simplicity}\label{sec:immunity-notion}

Intuitively, an {\em immune} language contains finite subsets and only infinite subsets that are ``hard'' to compute; in other words, it lacks any non-trivial ``easy'' subset. In contrast, a {\em simple} language inherits the immunity only for its complement. Such languages turn out to possess quite high complexity. The original notions of immunity and simplicity are rooted in the 1940s and later adapted to computational complexity theory in the 1970s with various restrictions on their computational resources. 

The notion of resource-bounded immunity for an arbitrary family $\CC$ of languages can be introduced in the following abstract way. A language $L$ is said to be {\em $\CC$-immune} if (i) $L$ is infinite and (ii) no infinite subset of $L$ exists in $\CC$. When a language family $\DD$ contains a $\CC$-immune language, we conveniently say 
that $\DD$ is {\em $\CC$-immune}.
Since $\CC$ cannot be $\CC$-immune, if $\DD$ is $\CC$-immune then it immediately follows that $\DD\nsubseteq\CC$. On the contrary, the separation $\DD\nsubseteq \CC$ cannot, in general, guarantee the existence of $\CC$-immune languages inside $\DD$. 
By this reason, a separation between two language families by immune languages is sometimes referred to as a {\em strong separation}. 
In a polynomial-time setting, for instance, even if assuming that $\p\neq\np$, it is not known whether there is a $\p$-immune language in $\np$ or equivalently $\np$ is $\p$-immune.

\subsection{Existence of Immune and Simple Languages}\label{sec:existence-immune}

Within a framework of formal language theory, we shall discuss the immunity of two well-known families of languages:  $\reg$ and $\cfl$. 
Earlier, Flajolet and Steyaert \cite{FS74} presented two examples: a $\reg$-immune language $L_{eq} = \{0^n1^n \mid n\in\nat\}$ and a $\cfl$-immune language $L_{3eq}=\{a^nb^nc^n\mid n\in\nat\}$. Notice that, in contrast, similar non-regular languages $Equal =\{x\in\{0,1\}^*\mid \#_0(x)=\#_1(x)\}$ and $3Equal =\{x\in\{0,1,2\}^*\mid \#_0(x)=\#_1(x)=\#_2(x)\}$ are not $\mathrm{REG}$-immune, because  
two regular languages $\{ (01)^n \mid n\in\nat \}$ and $\{ (012)^n \mid n\in\nat \}$ are respectively infinite subsets of $Equal$ and of $3Equal$. This clear contrast signifies a ``structural'' difference among 
those languages. We shall see more examples of immune languages. 

Since $\reg\subseteq \cfl$, the $\cfl$-immunity clearly implies the $\reg$-immunity but the converse does not hold because, for instance, $L_{eq}$ is $\reg$-immune and also belongs to $\cfl$. Since $L_{eq}$ and $L_{3eq}$ are {\em sparse} languages (because, \eg  $dense(L_{eq})(n) \leq 1$ for all lengths $n\in\nat$),  they  belong to the advised class $\reg/n$. Therefore, since  
$L_{eq}\in\dcfl$ and $L_{3eq}\in\cfl(2)$, the language family  $\dcfl\cap\reg/n$ is $\reg$-immune, and $\cfl(2)\cap\reg/n$ (thus $\cfl(2)\cap\cfl/n$) is $\cfl$-immune.  
In addition to these results, we remark that 
the language family $\mathrm{DCFL}-\reg/n$ is also $\reg$-immune.
A simple example is the ``marked'' language $Pal_{\#} =\{ w\#w^{R}\mid w\in\{0,1\}^*\}$ over the ternary alphabet $\{0,1,\#\}$, 
where $\#$ is used only as a {\em separator}. Notice that a use of this separator is crucial because a corresponding unmarked version $Pal = \{ww^{R}\mid w\in\{0,1\}^*\}$ (even-length palindromes) 
is no longer $\reg$-immune. 
The $\reg$-immunity of $\dcfl - \reg/n$ can be obtained simply by applying a standard pumping lemma for regular languages \cite{BPS61} (for the immunity of $Pal_{\#}$) and a swapping lemma for regular languages \cite{Yam08} (for the non-membership of $Pal_{\#}$ to $\reg/n$).  
When turning to the $\cfl$-immunity, on the contrary, 
it is not known whether $\cfl(2)-\cfl/n$ is $\cfl$-immune. The bast we can show at present is that $\mathrm{L}-\cfl/n$ is $\cfl$-immune, where $\mathrm{L}$ consists of all languages recognized by deterministic Turing machines with a single read-only input tape and a 
logarithmic-space bounded work tape.  
A typical example is the marked language 
$3Dup_{\#} = \{w\# w\# w\mid w\in\{0,1\}^*\}$. 
A standard pumping lemma for context-free languages \cite{BPS61} 
proves the $\cfl$-immunity of $3Dup_{\#}$; 
moreover, a direct use of a swapping lemma for context-free 
languages \cite{Yam08} 
proves that $3Dup_{\#}\not\in \cfl/n$. Since $3Dup_{\#}\in \mathrm{L}$, the $\cfl$-immunity of $\mathrm{L}-\cfl/n$ follows immediately.   

The immunity notion has given rise to the notion of {\em simplicity}. In general, a language $L$ is called {\em $\CC$-simple} if (i) $L$ is infinite, (ii) $L$ is in $\CC$, and (iii) $\overline{L}$ is $\CC$-immune. 
The existence of such a $\CC$-simple language clearly leads to a class separation $\CC\neq \co\CC$. {}Because of this implication, we do not know whether $\np$-simple languages exist (since, otherwise, $\np\neq\co\np$ follows). It is therefore natural to ask if $\cfl$-simple languages actually exist. In what follows, 
we prove the existence of such $\cfl$-simple languages.

\begin{proposition}\label{CFL-simple}
There exist $\cfl$-simple languages. Moreover, the complements of some of those languages belong to $\cfl(2)\cap\reg/n$.
\end{proposition}

Our example of $\cfl$-simplicity is the complement of a language $L_{keq}$ ($k\geq 3$), which is a natural generalization of $L_{3eq}$. Let $k\geq 3$ be fixed.  We define $L_{keq} = \{\sigma_1^n\sigma_2^n\cdots \sigma_k^n\mid n\in\nat\}$ over the $k$-letter alphabet $\Sigma_k =\{\sigma_1,\sigma_2,\ldots,\sigma_k\}$.  
We shall show that the complement of $L_{keq}$ is indeed $\cfl$-simple. This gives a clear contrast with the fact that both the language $3Equal$ 
(associated with $L_{3eq}$) and its complement are not even $\reg$-immune. 

\begin{proofof}{Proposition \ref{CFL-simple}}
Let $k$ be any integer at least $3$. 
We intend to show that (1) $\overline{L_{keq}}$ is in $\mathrm{CFL}$, (2) $L_{keq}$ is in $\cfl(2)\cap\reg/n$, and (3) $L_{keq}$ is $\cfl$-immune. 

(1) Our first claim is that $\overline{L_{keq}}$ belongs to $\mathrm{CFL}$. To simplify our proof, we shall argue only on the case $k=3$. Let us introduce two additional languages 
$L_{3} = \{\sigma_1^k\sigma_2^l\sigma_3^m \mid k,l,m\in\nat\}$ and 
$L_{3neq} = \{ \sigma_1^k\sigma_2^l\sigma_3^m \mid k\neq l, l\neq m, \text{ or } k\neq m\}$. Note that
$L_{3neq}$ equals the union of the following three sets: 
$\{\sigma_1^k\sigma_2^l\sigma_3^m \mid k\neq l,m\geq0\}$, 
$\{\sigma_1^k\sigma_2^l\sigma_3^m \mid l\neq m,k\geq0\}$, and 
$\{\sigma_1^k\sigma_2^l\sigma_3^m \mid m\neq k,l\geq0\}$, all of which are apparently context-free. 
Since $\mathrm{CFL}$ is closed under union, 
$L_{3neq}$ should belong to $\mathrm{CFL}$.
Moreover, since $\overline{L_{3eq}} = L_{3neq}\cup \overline{L_{3}}$
and $\overline{L_3}\in\mathrm{REG}\subseteq \mathrm{CFL}$, the language $\overline{L_{3eq}}$ is also in $\mathrm{CFL}$. 

(2) To show that $L_{keq}\in\reg/n$, choose an advice function $h$ 
defined as $h(n)=\sigma_1^{n/k}\sigma_2^{n/k}\cdots \sigma_{k}^{n/k}$ for all numbers $n\equiv 0$ ($\mathrm{mod}\;k$) and $h(n)=0^n$ for all the other $n$'s. 
If we define $S=\{\track{w}{w}\mid w\in\Sigma_k^*\}$,
then $\track{w}{h(|w|)}$ is in $S$ exactly when $w=h(|w|)$, which means  that $w\in L_{keq}$. Thus, $L_{keq}$ belongs to $\reg/n$. 
To show that $L_{keq}\in \cfl(2)$, let us deal only with the case where $k=2m$ and $m=2j+1$ for a certain number $j\in\nat^{+}$, since the other cases are similar. We introduce two useful languages $L_1$ and $L_2$ defined as follows: $L_1$ (resp., $L_2$) consists of all strings of the form $\sigma_1^{n_1}\sigma_2^{n_2}\cdots \sigma_{k}^{n_k}$ such that $n_{i}=n_{k+1-i}$ for all indices $i\in[1,m]_{\integer}$ (resp., $n_{2i+1}=n_{2i+2}$  and $n_{2i+m+1}=n_{2i+m+2}$ for all $i\in[0,j-1]_{\integer}$). 
Clearly, $L_1$ and $L_2$ are both context-free. Since the target language $L_{keq}$ can be expressed as $L_1\cap L_2$, $L_{keq}$ belongs to $\cfl(2)$. 

(3) Finally, we shall check the $\cfl$-immunity of $L_{keq}$. Assume that there exists an infinite subset $A\in\cfl$ of $L_{keq}$. To this $A$, we then apply a standard pumping lemma for context-free 
languages.\footnote{[Pumping Lemma for 
Context-Free Languages]\hs{1}
Let $L$ be any infinite context-free language. There exists a positive number $m$ such that, for any $w\in L$ with $|w|\geq m$, $w$ can be decomposed as $w = uvxyz$ with the following three conditions: (i) $|vxy|\leq m$, (ii) $|vy|\geq1$, and $uv^ixy^iz$ is in $L$ for any $i\in\nat$. See \cite{BPS61,HMU01}.}  
Let $m$ be a pumping-lemma constant. Choose $w=\sigma_1^n\sigma_2^n\cdots \sigma_k^n$ in $A$ with $n\geq m$. Take a decomposition $w=uvxyz$ with $|vxy|\leq m$ and $|vy|\geq 1$ such that $uv^jxy^jz$ is in $A$ for every index 
$i\in\nat$. Since $|vxy|\leq m\leq n$, there exists an index $i$ such that 
$vxy$ is a substring of either $\sigma_i^n$ or $\sigma_i^n\sigma_{i+1}^n$. Thus, 
we need to examine only two cases: (i) $v$ and $y$ are both substrings of $\sigma_i^n$ or (ii) $v$ is a substring of $\sigma_i^n$ and $y$ is a substring of $\sigma_{i+1}^n$. In either case, the string $uv^2xy^2z$ cannot belong to $A$.
This is absurd, and therefore $A$ does not exist. We thus reach the desired  conclusion of the $\cfl$-immunity of $L_{keq}$.
\end{proofof}

Notice that our $\cfl$-simple languages $\overline{L_{keq}}$ is not even $\reg$-immune because, for instance, the language $\overline{L_{3}}$ is an infinite regular subset of $\overline{L_{3eq}}$. This immediately raises a  natural question of 
whether there exist $\reg$-immune $\cfl$-simple languages.

\subsection{Properties of Immune Languages}\label{sec:properties}

Immune languages lack infinite subsets of certain complexity, 
and therefore, as we have presented in the previous subsection, 
they are of quite high complexity. To improve our understandings of 
the $\reg$-immunity, we wish to examine this notion 
by studying its relationships to three existing 
notions---nonregularity, quasireduction, and hardcore. 
The first notion relates to a nonregularity measure, which leads to 
another characterization of the $\reg$-immunity. The {\em nonregularity} $N_{L}(n)$ of a language $L$ at $n$ is the total number of equivalence classes  in $\Sigma^n/\equiv_{L}$, where the relation $\equiv_{L}$ is defined as:  $x\equiv_{L}y$ iff $\forall z\in\Sigma^*[xz\in L\IFF yz\in L]$.  

\begin{proposition}\label{non-regularity}
A language $L$ is $\reg$-immune iff $L$ is infinite and, for every infinite subset $A$ of $L$ and for every constant $c>0$, $N_{A}(n)>c$ holds for an infinite number of indices $n\in\nat$. 
\end{proposition}

This proposition is a natural extension of the so-called {\em Myhill-Nerode Theorem} \cite{HU79}, which bridges between the nonregularity and $\reg$. We include its proof for completeness. 

\begin{proofof}{Proposition \ref{non-regularity}}
(If -- part) We prove a contrapositive. Assume that $L$ has an infinite subset $A$ in $\reg$. Since $A\in\reg$, by the Myhill-Nerode Theorem, 
the cardinality of the set $\Sigma^*/\equiv_{A}$ is finite.  In other words, $N_{A}(n)$ is upper-bounded by a certain constant, 
which is not depending on $n$.

(Only If -- part) Let $L$ be $\reg$-immune. Assume that there are an infinite subset $A$ of $L$ and a constant $c>0$ for which $N_{A}(n)\leq c$ for all but finitely many $n\in\nat$. Let $\{A_1,A_2,\ldots,A_{c}\}$ denote all equivalence classes in $\Sigma^*/\equiv_{A}$. Take the lexicographically minimal string, say, $a_i$ from each set $A_i$. Consider a dfa $M$ with its transition function $\delta$ defined by: $\delta(i,\sigma)=j$ iff $a_i\sigma\equiv_{A} a_j$. The set of final states is $F=\{i\mid a_i\in A\}$. It is not difficult to check that $M$ indeed recognizes $A$. This implies that $A$ is regular, a contradiction against the $\reg$-immunity of $L$. 
\end{proofof}

Our notion of $\onedlin$-m-quasireduction gives the second 
characterization to the $\reg$-immunity. 
Let us recall from Section \ref{sec:notation} 
the partial function class $\oneflin(\fpartial)$. 
A {\em $\onedlin$-m-quasireduction} from $L$ to $A$ is a single-valued partial function $f$ that satisfies the following two conditions: for every string $x$, (i) when $f(x)$ is defined, $x\in L$ iff $f(x)\in A$ and (ii) $f$ is in $\oneflin(\fpartial)$.

\begin{lemma}
The language $L$ is $\reg$-immune iff $L$ is infinite and for any set $A$ and for any $\onedlin$-m-quasireduction $f:L\rightarrow A$ and for any $u\in A$, $f^{-1}(u)$ is finite. 
\end{lemma}

\begin{proof}
(If -- part) Assume that an infinite language $L$ is not $\reg$-immune. Take an infinite regular subset $A\subseteq L$. Choose an element 
$u_0\in A$ and, for every string $x$, define $f(x)=u_0$ if $x\in A$ and undefined otherwise. Since $f^{-1}(u_0)$ coincides with $A$,  
$f^{-1}(u_0)$ is infinite. Moreover, $f$ belongs to $\oneflin(\fpartial)$ since $A\in\reg$.  Thus, $f$ is a 
$\onedlin$-m-quasireduction from $L$ to $A$. 

(Only If -- part) Assume that we have an infinite set $L$, another set $A$, a $\onedlin$-m-quasireduction $f:L\rightarrow A$, and an element $u_0\in A$ such that $B=_{def}f^{-1}(u_0)$ is infinite. Since $f\in\oneflin(\fpartial)$, take a linear-time deterministic 1TM $M$ that computes $f$. Note that, for every input $x$, $x\in B$ iff $M(x)$ halts in an accepting state and outputs $u_0$. Hence, $B$ is in $\reg$. 
Therefore, $L$ has an infinite regular subset. 
\end{proof}

Next, we give the third characterization of the $\reg$-immunity using a notion of ``hardcore''; however, our definition of ``hardcore'' differs 
from a time-restricted definition of {\em (polynomial) hardcore} for polynomial-time bounded computation (see, \eg \cite{BDG88} for its definition). With a use of an npda, we rather impose a space 
restriction on the size of a stack used by the npda. 
To be more accurate, for any npda   $M=(Q,\Sigma,\Gamma,\delta,q_0,z,F)$, any constant $k\in\nat$, and any input string $x\in\Sigma^*$, we introduce the notation $M(x)_k$ defined as follows: (1) $M(x)_k = 1$ if there is an accepting path of $M$ on the input $x$ with stack size at most $k$; (2) $M(x)_k=0$ if all computation paths of $M$ on $x$ are rejecting paths with stack size at most $k$; and (3) $M(x)_k$ is {\em undefined} otherwise.
A context-free language $A$ is called a {\em $\reg$-hardcore} for a language $L$ if, for any constant $k\in\nat$ and any npda $M$ recognizing $A$, there exists a finite set $B\subseteq L$ such that $M(x)_k$ is undefined for all strings $x\in L - B$.  

\begin{proposition}
The following two statements are equivalent. Let $L$ be any infinite context-free language.
\begin{enumerate}\vs{-1}
\item The language $L$ is $\reg$-immune.
\vs{-2}
\item The language $L$ is a $\reg$-hardcore for $L$.
\end{enumerate}
\end{proposition}

\begin{proof}
(1 $\Rightarrow$ 2) We shall prove a contrapositive. Let $L$ be any infinite  context-free language. Assuming that $L$ is not a $\reg$-hardcore for $L$, we plan to prove that $L$ has an infinite regular subset. 
There exist a constant $k\in\nat$  and an npda $M$ with $L(M)=L$ such that, for every finite set $B\subseteq L$, $M(x)_k$ is defined (\ie $M(x)_k\in\{0,1\}$) for a certain input $x\in L-B$. Now, let us  introduce a new npda $N$ as follows: on input $x$, $N$ simulates $M$ on $x$ nondeterministically and, along each computation path, whenever its stack size exceeds $k$, it immediately rejects $x$. Consider the set $L(N)$ of all strings accepted by $N$. By the definition of $N$, it follows that $L(N)\subseteq L$. 

First, we claim that $L(N)$ is regular. Since $k$ is a fixed constant, we can express the entire content of the stack as a certain new internal state. Tracking down this state, we can simulate $N$ using a certain nondeterministic finite automaton (or nfa). This implies that $L(N)$ is regular.
Next, we claim that $L(N)$ is infinite. 
For every finite subset $B$ of $L$, a certain string $x\in L-B$ satisfies 
$M(x)_{k}\in\{0,1\}$; hence, $x\in L(N)$. {}From this property, we can conclude that $L(N)$ is infinite. Therefore, $L(N)$ is an infinite regular subset of $L$. 

(2 $\Rightarrow$ 1) We first assume that an infinite context-free language 
$L$ is not $\reg$-immune. This means that there exists a dfa $M$ for which $L(M)\subseteq L$ and $L(M)$ is infinite. Since $L$ is context-free, 
take an npda $N$ that recognizes $L$. Now, let us define a new npda $M'$ as follows: on input $x$, $M'$ splits its computation into two nondeterministic computation paths and then simulates $M$ and $N$ along these paths separately. Clearly, $L(M')=L(M)\cup L(N) = L$. 
Choose $k=1$ and consider $M'(x)_k$. For every string $x\in L(M)$, $M'(x)_k=1$ follows since $M$ is a dfa and uses no stack space. 
Let $B$ be any 
finite subset of $L$. Because $L(M)-B$ is infinite within $L$, there exists a string $x$ in $L-B$ for which $M'(x)_k=1$. This implies that $L$ cannot be a $\reg$-hardcore for $L$.
\end{proof}

\subsection{Complexity of Bi-Immune Languages}\label{sec:bi-immunity}

The existence of natural $\reg$-immune languages within $\cfl$ encourages us to search for much ``stronger'' immune languages in $\cfl$. One such candidate is  another variant of $\CC$-immunity, known as 
$\CC$-bi-immunity \cite{BS85}, where a language $L$ is {\em $\CC$-bi-immune} if $L$ and its complement $\overline{L}$ are both $\CC$-immune. For brevity, a language family $\DD$ is said to be {\em $\CC$-bi-immune} if there is a $\CC$-bi-immune language in $\DD$. 
In the literature, time-bounded bi-immunity has been known to be related to the notion of {\em genericity}, which corresponds to certain finite-extension diagonalization arguments (see, \eg \cite{AFH88,YS05} for its connection).

Is there any $\reg$-bi-immune language in $\cfl$? All the  examples of context-free $\reg$-immune languages shown in Section 
\ref{sec:existence-immune} appear to lack the $\reg$-bi-immunity property. 
Related to the open question on the existence of $\reg$-immune $\cfl$-simple languages, discussed in Section \ref{sec:existence-immune}, if $\cfl$ is not $\reg$-bi-immune, then no $\cfl$-simple language can be $\reg$-immune. Unfortunately, we are unable to answer the question at this point; instead, we shall prove that the language family 
$\mathrm{L}\cap\reg/n$ is $\reg$-bi-immune.

\begin{proposition}\label{reg-bi-immune}
The languages family $\mathrm{L}\cap\reg/n$ is $\reg$-bi-immune.
\end{proposition}

How can we prove this proposition? Balc{\'a}zar and Sch{\"o}ning \cite{BS85} employed a diagonalization technique to construct a $\p$-bi-immune language inside $\mathrm{EXP}$ (deterministic exponential-time class). Notice that any $\p$-bi-immune language constructed by such 
a diagonalization 
depends on how to enumerate all languages in $\p$. In our proof below, 
without requiring any enumeration of languages in $\reg$, 
we explicitly present 
two $\reg$-bi-immune languages. Our desired $\reg$-bi-immune languages are 
$L_{even}$ and $L_{odd}$ given as follows: 
\begin{itemize}\vs{-1}
\item $L_{even} = \{w\in\{0,1\}^*\mid \exists k\in\nat\,[2k<\log^{(2)}{|w|}\leq 2k+1]\} \cup\{\lambda\}\cup\{0,1\}^2$, 
and 
\vs{-2}
\item $L_{odd} = \{w\in\{0,1\}^*\mid \exists k\in\nat\,[2k+1<\log^{(2)}{|w|}\leq 2k+2]\} \cup\{0,1\}$. 
\end{itemize}
Notice that these two languages form a partition of $\{0,1\}^*$; namely, 
$L_{even}\cup L_{odd}=\{0,1\}^*$ 
and $L_{even}\cap L_{odd}=\emptyset$. 

\begin{proofof}{Proposition \ref{reg-bi-immune}}
It suffices to show that $L_{even}$ and $L_{odd}$ are both $\reg$-immune because each of them is the complement of the other.  
For brevity, let $\Sigma=\{0,1\}$. We begin with proving the $\reg$-immunity of $L_{even}$ by contradiction.  
Assume that there exists an infinite regular subset $A$ of $L_{even}$.
 We apply to $A$ a standard pumping lemma  for regular 
languages.\footnote{[Pumping Lemma for Regular Languages]\hs{1}
Let $L$ be any infinite regular language. There exists a number  $m>0$ (referred to as a pumping-lemma constant) such that, for any string $w$ of length $\geq m$ in $L$, there is a decomposition $w=xyz$ for which (i) $|xy|\leq m$, (ii) $|y|\geq 1$, and (iii) $xy^iz\in L$ for any $i\in\nat$. See \cite{BPS61,HMU01}.}  
Take a pumping-lemma constant $m>0$ and then choose a string $w$ in $A\cap \Sigma^n$ for a certain length $n$ with $n\geq m+1$. Such $n$ satisfies that 
$2k<\log^{(2)}{n}\leq 2k+1$ for a certain number $k\in\nat$. The pumping lemma provides a decomposition $w=xyz$ with $|xy|\leq m$ and $|y|\geq1$ for which 
$w_i =_{def} xy^iz$ belongs to $A$ for any number $i\in\nat$.
Now, let $\ell =|y|$. 
Toward a contradiction, there are two cases to consider separately.

{\bf Case 1:} Consider the case where $\log^{(2)}{n} = 2k+1$. 
In this case, we choose $i=n+1$. 
Since $1\leq \ell\leq m$ and $m+1\leq n$, 
the length $|w_i|$ is sandwiched by two terms as 
\[
2^{2^{2k+1}}=n <|w_i| = n + (i-1)\ell \leq n + n\ell \leq n(m+1)
 \leq n^2 = 2^{2^{2k+2}}.
\] 
In short, it holds that 
$2k+1<\log^{(2)}{|w_i|}\leq 2k+2$, implying that $w_i$ is in $L_{odd}$. 
Since $A\cap L_{odd}=\emptyset$, it immediately follows that $w_i\not\in A$, a contradiction. 

{\bf Case 2:} Consider the case where $2k< \log^{(2)}{n}<2k+1$. This means that $2^{2^{2k}}<n\leq 2^{2^{2k+1}}-1$. When we choose $i = \ceilings{n(n-1)/\ell}+1$, the length $|w_i|$ can be lower-bounded by
\[
|w_i| = n+(i-1)\ell \geq n + \frac{n(n-1)}{\ell}\cdot \ell = n + n(n-1) = n^2 > 2^{2^{2k+1}}. 
\] 
In contrast, since $n\geq m+1 > m/2$, we can upper-bound $|w_i|$ as
\[
|w_i| < n + \left( \frac{n(n-1)}{\ell} + 1 \right)\cdot \ell = n^2 + \ell \leq n^2 + m < (n+1)^2 \leq 2^{2^{2k+2}}.
\]
These two bounds together imply that $2k+1<\log^{(2)}{|w_i|}<2k+2$, concluding that $w_i\in L_{odd}$, a contradiction
against the fact that $w_i\in A\subseteq L_{even}$.  

{}From the above two cases, we can conclude that $A$ does not exist; in other words, $L_{even}$ is $\reg$-immune, as requested.
Similarly, we can show that $L_{odd}$ is $\reg$-immune. 

We still need to argue that $L_{even}$ and $L_{odd}$ are both in $\mathrm{L}\cap\reg/n$. 
Since $\mathrm{L}\cap\reg/n$ is closed under complementation, 
it suffices to show that 
$L_{even}$ belongs to $\mathrm{L}\cap \reg/n$. 
First, we shall demonstrate that $L_{even}\in\reg/n$. Let us consider the following advice function $h(n)=10^{n-1}$ if $L_{even}\cap\Sigma^n\neq \emptyset$, and $h(n)=0^{n}$ if $L_{odd}\cap\Sigma^n\neq\emptyset$ for any length $n\geq1$; in addition, set $h(0)=\lambda$. Define a set $A$ as $A=\left\{\track{x}{1y}\mid |x|=|y|+1,y\in\{0,1\}^*\right\}$. 
It is obvious that, for every $x$, $x\in L_{even}$ iff $\track{x}{h(|x|)}\in A$.  
Since $A$ is regular, $L_{even}$ therefore belongs to $\reg/n$.
To show that $L_{even}\in\mathrm{L}$, let us consider the following algorithm for $L_{even}$. 
\begin{quote}
On input $x$, if $x=\lambda$ then accept it. Assume that $|x|\geq 1$. 
With access to $w$ written on a read-only input tape, compute $\ceilings{\log^{(2)}{|w|}}$ on its work tape.  If $\ceilings{\log^{(2)}{|w|}}$ is odd, then accept the input; otherwise, reject it. 
\end{quote}
It is not difficult to show that this algorithm recognizes $L_{even}$ using only logarithmic space. This completes our proof of the proposition. 
\end{proofof}

\section{P-Denseness and Primeimmunity}\label{sec:p-dense-immune}

We begin with a brief discussion on a density issue of $\reg$-immune languages. 
Recall that non-immunity of a language guarantees the existence of a certain infinite subset that is ``computationally easy.'' 
In many cases, these infinite subsets are of {\em low density}. 
In typical examples, there are infinite {\em sparse} subsets $\{(01)^n\mid n\in\nat\}$ and $\{(012)^n\mid n\in\nat\}$ inside $Equal$ and $3Equal$, respectively. 
Notice that all context-free $\reg$-immune languages $L$ described in Section \ref{sec:immunity-notion} satisfy the following density property: its density rate $dense(L)(n)/|\Sigma^n|$ is ``exponentially small'' in terms of a length parameter $n$. The language $Pal_{\#}$, for example, satisfies that $dense(Pal_{\#})(n)/|\Sigma^n|\leq 2^{\floors{n/2}}/3^n$ (thus $dense(Pal_{\#})(n)\leq|\Sigma^n|/(2.2)^n$) for every odd length $n\geq1$.  Naturally, we can question whether there exists a context-free $\reg$-immune language whose density rate is ``polynomially large.'' To be more precise, we call a language $L$ over an alphabet $\Sigma$ {\em polynomially dense} (or {\em p-dense}, in short) exactly when there exist a number $n_0\in\nat$ and a non-zero polynomial $p$ such that 
$dense(A)(n)\geq |\Sigma^n|/p(n)$ for all numbers $n\geq n_0$. Our previous question is now rephrased as: is there any p-dense $\reg$-immune language in $\cfl$, or is $\cfl$ p-dense $\reg$-immune?
It appears that we are unable to settle this question at present. This situation seems to signify the meaningfulness of the notion 
of p-denseness in our study of immunity. Meanwhile, we shall show that $\mathrm{L}\cap\cfl/n$ is indeed p-dense $\reg$-immune. 

\begin{proposition}\label{dense-immune}
The language family $\mathrm{L}\cap\cfl/n$ is p-dense $\reg$-immune.
\end{proposition}

Let us consider the language $LCenter =\{au0^m10^mv\mid a\in\{\lambda,0,1\}, 2^m\leq |u|=|v|<2^{m+1}\}$ over the alphabet $\{0,1\}$. Notice that $LCenter$ is in $\mathrm{L}\cap \cfl/n$. 
We claim in the following proof that this language is $\reg$-immune and also p-dense. 

\begin{proofof}{Proposition \ref{dense-immune}}
We want to show that $LCenter$ is p-dense $\reg$-immune. 
We first show that $LCenter$ is p-dense. Let $w=au0^m10^mv$ in $LCenter$
with $2^{m}\leq |u|=|v| < 2^{m+1}$. Let $n=|w|$. 
Consider the case where $a=\lambda$. 
In this case, since $2^{m}\leq |u| = (n-2m-1)/2<2^{m+1}$, 
we obtain $2^{m+1}+2m+1 \leq n$, which implies $n^2\geq 2^{2m+1}$. 
Since $dense(LCenter)(n) = 2^{n-2m-1}$, the density rate 
$\frac{dense(LCenter)(n)}{|\Sigma^n|}$ equals  
$\frac{1}{2^{2m+1}}$, which is clearly at least $1/n^2$.  
The other cases where $a\in\{0,1\}$ are similar. 
Therefore, $LCenter$ is p-dense.

Next, we show that $LCenter$ is $\reg$-immune. Assuming otherwise, 
we choose
an infinite subset $A$ of $LCenter$ in $\reg$. As in the proof of Proposition \ref{reg-bi-immune}, we use the pumping lemma for regular languages. 
Take a pumping-lemma constant $m>0$. Let $w = au0^k10^kv$ be any 
string in $A$ with $k>m$ and $2^{k}\leq |u|=|v| <2^{k+1}$. 
Now, assume that $a=\lambda$. The other cases are similar.
Let us take any decomposition $w = xyz$ with $|xy|\leq m$ and $|y|\geq1$ 
such that $xy^iz$ is in $A$ for any number $i\in\nat$.  
Since $|xy|\leq m<k$, $y$ is a substring of $u$. Consider the string 
$xz$.  
Clearly, the center symbol of $xz$ should be $0$. Thus, $xz$ cannot belong to $LCenter$. This is a contradiction against the fact that $xz\in A$. Therefore, $LCenter$ must be $\reg$-immune.
\end{proofof} 

Apart from the $\reg$-immunity, we turn our attention to 
p-dense languages 
that lack only p-dense regular subsets. 
Such languages are referred to as $\reg$-primeimmune. 
More generally, for a language family $\CC$, we say that a language $L$ over $\Sigma$ is {\em $\CC$-primeimmune} if (1) $L$ is p-dense and (2) $L$ has no p-dense subset in $\CC$. A language family $\DD$ is {\em $\CC$-primeimmune} if there exists a $\CC$-primeimmune language in $\DD$. This definition immediately yields, similar to the $\CC$-immunity,  the self-exclusion property: $\CC$ cannot be $\CC$-primeimmune. 

The following obvious relationship holds between p-dense $\reg$-immunity and $\reg$-primeimmunity. If a language $L$ is p-dense but not $\reg$-primeimmune, then $L$ contains a p-dense regular subset, say,  $A$. 
By the definition of p-denseness, $A$ should be infinite and thus $L$ must not be $\reg$-immune. The next lemma therefore follows.

\begin{lemma}\label{dense-immune-prime}
Let $L$ be any language over an alphabet $\Sigma$ 
with $|\Sigma|\geq2$. If $L$ is p-dense $\reg$-immune, 
then $L$ is $\reg$-primeimmune.
\end{lemma}

Although $\cfl$ is not known to be p-dense $\reg$-immune, it is 
possible for us to show that $\cfl$ is $\reg$-primeimmune. First, recall the context-free language $Equal$ over the binary alphabet $\{0,1\}$.  Since $Equal$ is technically not p-dense, we need to extend
it slightly and define its ``extended'' language $Equal_{*}$ as 
$\{aw\mid a\in\{\lambda,0,1\}, w\in Equal\}$.
Despite $Equal_{*}$'s non-$\reg$-immunity,  
we can prove that $Equal_{*}$ is $\reg$-primeimmune. In the next proposition,  we shall challenge a slightly stronger statement: $Equal_{*}$ is $\reg/n$-primeimmune. This highlights a stark difference between the 
$\reg/n$-primeimmunity and the $\reg/n$-immunity, since there exists no $\reg/n$-immune language (because every infinite language $L$ over an alphabet $\Sigma$ has an infinite subset of the form $\{\sigma x\in L\mid \sigma\in\Sigma, x\in\Sigma^*, h(|\sigma x|)=\tilde{\sigma} x\}$ in $\reg/n$, where $\tilde{\sigma} = \track{\sigma}{1}$ and $h$ is an advice function defined as $h(n)=\tilde{\sigma}x$ if $\sigma x$ is the lexicographically minimal string in $L\cap\Sigma^n$ and $h(n)=0^{n}$ otherwise). 

\begin{proposition}\label{equal-dense-immune}
The language $Equal_{*}$ is $\reg/n$-primeimmune. 
\end{proposition}

\begin{proof}
We start our proof with an easy claim on the p-denseness of $Equal_{*}$. 
For any sufficiently large even number $n$, by Stirling's approximation formula, the density of $Equal_{*}$ can be estimated as
\begin{equation}\label{Stirling}
dense(Equal_{*})(n) = \comb{n}{n/2} = \frac{2^n\sqrt{2}}{\sqrt{\pi n}}\left(1+ \Theta\left(\frac{1}{n}\right)\right) > \frac{2^n}{n}.
\end{equation}
When $n$ is odd, on the contrary, since  
$dense(Equal_{*})(n)$ equals $2\cdot dense(Equal_{*})(n-1)$, 
it is upper-bounded by $\frac{2\cdot2^{n-1}}{n-1} > 2^n/n$ with 
a help of Eq.(\ref{Stirling}).
These two lower bounds yield the desired p-denseness of $Equal_{*}$.

Our next goal is to prove the non-existence of p-dense subset of $Equal_{*}$ in $\reg/n$. Assume otherwise; namely, there is a p-dense set $A\subseteq Equal_{*}$ in $\reg/n$. Since $A$ is p-dense, a certain constant $d\geq1$ satisfies $dense(A)(n)\geq 2^n/n^d$ for all but finitely many numbers $n$.  Here, we shall apply a swapping lemma for regular languages.\footnote{[Swapping Lemma for Regular Languages] \hs{1}
Let $L$ be any infinite regular language on an alphabet 
$\Sigma$ with $|\Sigma|\geq2$. 
There exists a positive integer $m$ such that, for any integer $n\geq 1$ and any subset $S$ of $L\cap\Sigma^n$ of cardinality at least $m$, the following condition holds: for any integer $i\in[0,n]_{\integer}$, there exist two strings $x=x_1x_2$ and $y=y_1y_2$ in $S$ with $|x_1|=|y_1|=i$ and $|x_2|=|y_2|$ satisfying that (i) $x\neq y$, (ii) $y_1x_2\in L$, and (iii)  $x_1y_2\in L$. See \cite{Yam08}.} 
Let $m$ be a swapping-lemma constant for $A$ and choose a sufficiently large number $n$ in $\nat$. 
It suffices to consider only the case where  
$m$ is odd. Without loss of generality, we further assume that $m\geq 5$. 
For each pair $i,k\in[0,n]_{\integer}$, 
the notation $A_{k,i}$ denotes the set $\{x \in A\cap\Sigma^n \mid \#_0(pref_{k}(x)) = i\}$ so that $A\cap\Sigma^n$ can be expressed as $A\cap\Sigma^n = \bigcap_{k=0}^{n}\left(\bigcup_{i=0}^{n}A_{k,i}\right)$.
Now, we state a key property of $\{A_{k,i}\}_{k,i}$, from which the desired proposition immediately follows. 

\begin{claim}\label{m-indices}
There are an index 
$k\in[m-1,n]_{\integer}$ and at least $m$ distinct indices $(i_1,i_2,\ldots,i_m)$ such that $A_{k,i_j}\neq \emptyset$ for every index $j\in[1,m]_{\integer}$.  
\end{claim}

Assuming that Claim \ref{m-indices} is true, let us choose an index $k\in[m-1,n]_{\integer}$ and 
$m$ distinct indices $(i_1,\ldots,i_m)$ that satisfy the claim. We then choose one string $w_j$ from each set $A_{k,i_j}$ and define $W=\{w_1,w_2,\ldots,w_m\}$. Since $|W|\geq m$, by the swapping lemma, there are two distinct strings $x_1x_2$ and $y_1y_2$ in $W$ with $|x_1|=|y_1|$ and $|x_2|=|y_2|$ such that the swapped strings $x_1y_2$ and $y_1x_2$ belong to $A$. This leads to 
a contradiction because the choice of $W$ makes $x_1y_2$ satisfy $\#_0(x_1y_2)\neq \#_1(x_1y_2)$. This contradiction leads us to conclude that $A$ does not exist, and therefore we finish the proof of Proposition \ref{equal-dense-immune}.

Now, our remaining task is to prove Claim \ref{m-indices}. 
Assume that this claim is false; that is, 
(*) for each index $k\in[m-1,n]_{\integer}$, there are at most $m$ indices, say, $(i_1,\ldots,i_{m'})$, where $m'\leq m$, satisfying $A_{k,i_j}\neq \emptyset$ for all indices  $j\in[1,m']_{\integer}$. For convenience, we write $I_{k}^{*}$ for the set $\{i_1,\ldots,i_{m'}\}$ of such indices.  
In the rest of the argument, we abbreviate  $\ceilings{m/2}$ as $m_0$ for brevity. Note that $2m_0=m+1$. 
Since $m$ is fixed, we often omit ``$m_0$'' and ``$m$.'' 

Toward a contradiction, we intend to estimate the value $|A\cap\Sigma^n|$. Since $A$ is p-dense, we can obtain a lower bound $|A\cap\Sigma^n|\geq 2^n/n^{d}$ for all but finitely many numbers $n$. In contrast, the following statement gives an upper bound of $|A\cap\Sigma^n|$. 

\begin{claim}\label{bound-A}
There exists a constant $c$, depending only on $m$, with $1<c<2$ satisfying 
that $|A\cap\Sigma^n| <c^n$ for all sufficiently large numbers $n$. 
\end{claim}

Together with the p-denseness of $A$, Claim \ref{bound-A} yields a relation  $2^n/n^{d}\leq |A\cap\Sigma^n|<c^n$, 
from which we immediately obtain $c>2n^{-d/n}$. Since 
$\lim_{n\rightarrow\infty}n^{-d/n}=1$, we reach a conclusion $c\geq 2$, which clearly contradicts the choice of $c$ in Claim \ref{bound-A}. 
Therefore, Claim \ref{m-indices} holds. 

To complete the proof of our proposition, we need to prove Claims 
\ref{bound-A}. For this purpose, let us consider all possible sets $A$ that satisfy Condition (*) stated above and 
let $\AAA$ denote the collection of all such sets.    
Now, we want to discuss what kind of $A\in\AAA$ 
gives $|A\cap\Sigma^n|$ the largest value.  
Here is an explicit candidate for such $A$'s. 
Let $k\geq m-1$. We first define the integer interval $I_k=[\ceilings{(k+1)/2}-(m_0-1),\ceilings{(k+1)/2}+(m_0-1)]_{\integer}$ (whose center point is $\ceilings{(k+1)/2}$) of size $m$; 
in particular, $I_{m-1}=[1,m]_{\integer}$. 
Next, we introduce $S_k$ as the set of all strings $w\in\Sigma^k$ 
such that, for each index $j\in[m-1,k]_{\integer}$,  
$\#_{0}(pref_j(w))$ belongs to $I_j$. 
The set $S=_{def}\bigcup_{k\in\nat} S_{k}$ clearly falls into $\AAA$.

In what follows, we shall claim that  (1) $|S_n|$ is at most $c^n$ for a certain constant $c$ with $1<c<2$ and 
(2)  for every set $A\in\AAA$,  $|S_{n}|$ upper-bounds $|A\cap\Sigma^n|$.   
These form the core of our proof.   
We begin with the first claim by making a direct estimation of the target value $|S_n|$.

\begin{claim}\label{upper-bound-c}
There exists a constant $c$, depending only on $m$, with $1<c<2$ such that $|S_{n}| <c^n$ for all sufficiently large numbers $n\in\nat$. 
\end{claim}

\begin{proof} 
\sloppy Recall that $m$ is an odd number at least $5$. 
To estimate each value $|S_e|$, where $m-1\leq e\leq n$, 
we first partition $S_e$ into $S_{e,1},S_{e,2},\ldots,S_{e,m}$, where $S_{e,i} = \{w\in S_e\mid \#_{0}(w) \text{ is the $i$th element in $I_e$ }\}$ for any index $i\in[1,m]_{\integer}$. Note that ``$w\in S_{e,i}$'' yields the equation $\#_{0}(pref_{e}(w))=\ceilings{(e+1)/2}-m_0+i$. For convenience, we write $a_{e,i}$ to denote the cardinality $|S_{e,i}|$. 
A simple observation provides the following relations among $S_{e,i}$'s: if $e$ is odd, then $S_{e,i} =\{w0\mid w\in S_{e-1,i}\}\cup \{w1\mid w\in S_{e-1,i+1}\}$; otherwise, $S_{e,i} =\{w0\mid w\in S_{e-1,i-1}\} \cup \{w1\mid w\in S_{e-1,i}\}$, where we assume that $S_{e-1,m+1}=S_{e-1,0}=\emptyset$. 
In the rest of this proof, we are focused only on odd values of $e$.

The aforementioned relations among $S_{e,i}$'s imply that, for any index $k\in[1,(m-3)/2]_{\integer}$,   
\begin{equation}\label{recurrence-a} 
a_{2k+3,1} = 2a_{2k+1,1} + a_{2k+1,2}, \;\;
a_{2k+3,m} = a_{2k+1,m-1} + a_{2k+1,m}, \;\; \text{and} \vs{-2}
\end{equation}
\[
a_{2k+3,i} = a_{2k+1,i-1} + 2a_{2k+1,i} + a_{2k+1,i+1}. \hs{45}
\]

\n Notice that $a_{2k+3,m}$ is the smallest and $a_{2k+3,1}$ is the second smallest among $a_{2k+3,i}$'s. Since $|S_{2k+3}| = \sum_{1\leq i\leq m} |S_{2k+3,i}|$, from Eq.(\ref{recurrence-a}), it follows that   
\[
|S_{2k+3}|  
= 3a_{2k+1,1} + 2a_{2k+1,m} + 4\sum_{2\leq i\leq m-1} a_{2k+1,i} \leq 3|S_{2k+1}| + \sum_{2\leq i\leq m-1} a_{2k+1,i}. 
\]
To calculate $|S_{2k+3}|$, we thus need to estimate the sum $\sum_{2\leq i\leq m-1}a_{2k+1,i}$ in terms of $|S_{2k+1}|$. Our starting point is the following simple upper bound of $\sum_{2\leq i\leq m-1}a_{2k+1,i}$ by a certain constant multiple of $a_{2k+3,1}+a_{2k+3,m}$. 

\begin{claim}\label{delta-inequality}
It holds that  
$\sum_{i=m_{0}-j+1}^{m_{0}+j-1} a_{2k+3,i} \leq 
\delta_{j}(a_{2k+3,m_{0}-j} + a_{2k+3,m_{0}+j})$  
for each index $j\in[1,m_{0}-1]_{\integer}$, where $\delta_j = 2^{2j-1}-1$. 
In particular, $\sum_{i=2}^{m-1}a_{2k+3,i} \leq \delta_{m_{0}-1}(a_{2k+3,1} + a_{2k+3,m})$. 
\end{claim}

\begin{proof}
For notational succinctness, we write $b_{e,j}$ for $a_{e,m_0-j}+a_{e,m_0+j}$. Now, we want to show by induction on $j$ 
that $\sum_{i=m_{0}-j+1}^{m_{0}+j-1} a_{2k+3,i} \leq 
\delta_{j}b_{2k+3,j}$. Consider the basis case $j=1$. 
By Eq.(\ref{recurrence-a}), it follows that
\begin{eqnarray*}
b_{2k+3,1} &=& a_{2k+1,m_0-2} + 2(a_{2k+1,m_0-1} + a_{2k+1,m_0} + a_{2k+1,m_0+1})  + a_{2k+1,m_0+2} \\
&\geq&  a_{2k+1,m_0-1} + 2a_{2k+1,m_0} + a_{2k+1,m_0+1}  
\;\;=\;\; a_{2k+3,m_0}. 
\end{eqnarray*}
This inequality yields the desired relation $a_{2k+3,m_0}\geq \delta_1 b_{2k+3,1}$ since $\delta_1=1$. 

Let us consider the induction step $j$ with $2\leq j\leq m_0-1$.
We first discuss the case where $j\neq m_0-1$. 
Note that the sum $\sum_{i=m_{0}-j+1}^{m_{0}+j-1}a_{2k+3,i}$ equals   
\begin{eqnarray*}
a_{2k+1,m_0-j} + 3a_{2k+1,m_0-j+1} + 4\sum_{i=m_{0}-j+2}^{m_{0}+j-2} a_{2k+1,i} + 3a_{2k+1,m_0+j-1} + a_{2k+1,m_0+j}.
\end{eqnarray*}
The induction hypothesis on $j-1$ yields $\sum_{i=m_{0}-j+2}^{m_{0}+j-2} a_{2k+1,i} \leq \delta_j b_{2k+1,j-1}$. With a help of this inequality, the sum $\sum_{i=m_{0}-j+1}^{m_{0}+j-1}a_{2k+3,i}$ is bounded from above by  
\begin{eqnarray*}
\sum_{i=m_{0}-j+1}^{m_{0}+j-1} a_{2k+3,i}  
&\leq& a_{2k+1,m_0-j} + (4\delta_{j-1}+3) (a_{2k+1,m_0-j+1} + a_{2k+1,m_0+j-1}) 
+ a_{2k+1,m_0+j}.
\end{eqnarray*}
Moreover, Eq.(\ref{recurrence-a}) gives a lower bound of  $b_{2k+3,j}$ as follows:
\begin{eqnarray*}
b_{2k+3,j} &\geq& 2a_{2k+1,m_0-j} + a_{2k+1,m_0-j+1} + a_{2k+1,m_0+j-1} + 2a_{2k+1,m_0+j}. 
\end{eqnarray*}
We therefore obtain the bound $\sum_{i=m_{0}-j+1}^{m_{0}+j-1} a_{2k+3,i} \leq (4\delta_{j-1}+3) b_{2k+3,j}$. Since $\delta_{j}$ satisfies that $\delta_{j} = 4\delta_{j-1}+3$, the desired relation immediately follows. The case where $j=m_0-1$ is treated similarly with a minor modification. 
By applying the induction, we obtain the claim.  
\end{proof}

By Claim \ref{delta-inequality}, $|S_{2k+1}|$ is lower-bounded by 
\[
|S_{2k+1}| = a_{2k+1,1} + a_{2k+1,m} + \sum_{2\leq i\leq m-1}a_{2k+1,i} \geq \left(\frac{1}{\delta_{m_0-1}}+1\right) \sum_{2\leq i\leq m-1}a_{2k+1,i},
\]
from which we obtain $\sum_{2\leq i\leq m-1}a_{2k+1,i} \leq \gamma |S_{2k+1}|$ if we set $\gamma =1/(1/\delta_{m_0-1}+1)<1$. We therefore conclude that 
$
|S_{2k+3}| \leq 3|S_{2k+1}| + \sum_{2\leq i\leq m-1}a_{2k+1,i} \leq (3+\gamma)|S_{2k+1}|.
$  
This recurrence has a solution $|S_{n}|\leq (3+\gamma)^{(n-m)/2}|S_{m}|$ for every {\em odd} number $n\geq m$. Since $|S_{2k+2}|\leq |S_{2k+3}|$ and $m\geq5$, it holds that $|S_n|\leq (3+\gamma)^{n/2}|S_m|$ for {\em all} 
numbers $n\geq1$. 
In this end, the fact that $|S_m|$ is a constant 
and $1<(3+\gamma)^{1/2}\ <2$ 
leads to Claim \ref{upper-bound-c}. 
\end{proof}

Finally, we want to prove the second claim that 
$|A\cap\Sigma^n|\leq |S_n|$. In Claim \ref{A-vs-S}, 
we actually prove a much stronger statement. 
To describe this claim, we shall explain  
new terminology.  
Let $k\in[m-1,n]_{\integer}$ be an arbitrary number. 
A {\em convergence point} is an $m$-tuple $(d_1,d_2,\ldots,d_m)$ that satisfies the following conditions: (i) for all indices $i\in[1,m]_{\integer}$, $d_i$ is in $\nat$ and 
(ii) $d_1\leq d_2\leq \cdots \leq d_m$.  
For any two convergence points $(d_1,d_2,\ldots,d_m)$  and $(d'_1,d'_2,\ldots,d'_m)$, we say that $(d_1,d_2,\ldots,d_m)$ {\em majorizes} $(d'_1,d'_2,\ldots,d'_m)$ if, for every index $k\in[1,m]_{\integer}$,  
$\sum_{k\leq i\leq m} d_i\geq \sum_{k\leq i\leq m} d'_i$. 
This majorization notion directly implies that $\sum_{1\leq i\leq m} d_i\geq \sum_{1\leq i\leq m} d'_i$. 

Let us recall that $a_{k,i}$ denotes $|S_{k,i}|$. Among $a_{k,i}$'s, the following relation holds: when $k$ is odd, $a_{k,m}\leq a_{k,1}\leq a_{k,m-1}\leq a_{k,2} \leq\cdots\leq a_{k,m_0}$ and, when $k$ is even, $a_{k,1}\leq a_{k,m}\leq a_{k,2}\leq a_{k,m-1} \leq\cdots\leq a_{k,m_0}$. 
To simplify the description of $a_{k,i}$'s in these enumerations, we introduce another notation $\tilde{a}_{k,i}$ to denote the 
$i$th element in the corresponding enumeration; thus, 
for {\em every} index $k$, $\tilde{a}_{k,1}\leq \tilde{a}_{k,2}\leq \cdots \leq\tilde{a}_{k,m}$. 
The $m$-tuple $(\tilde{a}_{k,1},\tilde{a}_{k,2},\ldots,\tilde{a}_{k,m})$ becomes a convergence point. 
It is not difficult to show by induction that, for any index $i\in[1,m]_{\integer}$, $\tilde{a}_{k,i} = \tilde{a}_{k-1,i-1}+\tilde{a}_{k-1,i+1}$, where we conveniently set $\tilde{a}_{k-1,0}=0$ and $\tilde{a}_{k-1,m+1}=\tilde{a}_{k-1,m}$.  

Associated with $A_{k,i}$, we introduce another 
notation $A^{*}_{k,i}$, analogous to $S_{k,i}$'s, to denote the set  
$\{pref_{k}(x)\mid x\in A\cap\Sigma^n, \#_{0}(pref_{k}(x))=i\}$ 
and let $A^{*}_{k} = \bigcup_{i\in I^{*}_{k}} A^{*}_{k,i}$.  
Without loss of generality, we can assume that $|I_k^*|=m$ (because,  otherwise, we add appropriate elements to $I^{*}_{k}$). 
In general, there may be a situation in which $w_1,w_2\in A^{*}_{k-1,i}$ and $w_1b\in A^{*}_{k,j}$ but $w_2b\not\in A^{*}_{k,j}$ for certain elements $w_1,w_2,b$. Clearly, this situation decreases the value $|A^{*}_{k,j}|$; hereafter, it suffices to 
assume that this situation never occurs. 

\sloppy To simplify our description in the following argument, 
we enumerate all $A^{*}_{k,i}$'s as $B_{k,j}$'s so that  
$|B_{k,1}|\leq |B_{k,2}|\leq \cdots\leq |B_{k,m}|$. Obviously,  $(|B_{k,1}|,|B_{k,2}|,\ldots,|B_{k,m}|)$ 
becomes a convergence point.  
Toward 
the desired result  $|A\cap\Sigma^n|\leq |S_n|$, since  $|A\cap\Sigma^n| = \sum_{1\leq i\leq m}|B_{k,i}|$ and $|S_n| = \sum_{1\leq i\leq m}\tilde{a}_{k,i}$, it is enough to show that 
$(\tilde{a}_{k,1},\tilde{a}_{k,2},\ldots,\tilde{a}_{k,m})$ majorizes $(|B_{k,1}|,|B_{k,2}|,\ldots,|B_{k,m}|)$. 

\begin{claim}\label{A-vs-S}
Let $k\in[m-1,n]_{\integer}$ and let $A\in\AAA$.  Consider 
$A^{*}_{k,1},A^{*}_{k,2},\ldots,A^{*}_{k,m}$ induced from $A\cap\Sigma^n$ as described before. 
Let $B_{k,1},B_{k,2},\ldots,B_{k,m}$ be an enumeration of 
$A^{*}_{k,i}$'s so that 
$|B_{k,1}|\leq |B_{k,2}|\leq \cdots\leq |B_{k,m}|$.  
It then holds that $(\tilde{a}_{k,1},\tilde{a}_{k,2},\ldots,\tilde{a}_{k,m})$ majorizes $(|B_{k,1}|,|B_{k,2}|,\ldots,|B_{k,m}|)$.  Thus, in particular, $|A\cap\Sigma^n|\leq |S_n|$ holds. 
\end{claim}

Our proof of Claim \ref{A-vs-S} 
is comprised of two extra claims---Claim \ref{preserve-maj} and \ref{max-choice}.

\begin{claim}\label{preserve-maj}
Let $(d_1,d_2,\ldots,d_m),(c_1,c_2,\ldots,c_m)$ be any two convergence points.  For every index $i\in[1,m]_{\integer}$, define $\tilde{d}_i = d_{i-1}+ d_{i+1}$ with $d_0=0$ and $d_{m+1}=d_{m}$ and define $\tilde{c}_i = c_{i-1}+ c_{i+1}$ with $c_0=0$ and $c_{m+1}=c_{m}$. If $(d_1,d_2,\ldots,d_m)$ majorizes $(c_1,c_2,\ldots,c_m)$, then $(\tilde{d}_1,\tilde{d}_2,\ldots,\tilde{d}_m)$ majorizes $(\tilde{c}_1,\tilde{c}_2,\ldots,\tilde{c}_m)$.
\end{claim}

\sloppy Since the proof of this claim is rather short, we shall give it here.  
Let $k$ be any index in $[1,m]_{\integer}$. 
Since $(d_1,\ldots,d_m)$ majorizes $(c_1,\ldots,c_m)$, it holds 
that $\sum_{k-1\leq i\leq m} d_i \geq \sum_{k-1\leq i\leq m}c_i$ and  $\sum_{k+1\leq i\leq m} d_i \geq \sum_{k+1\leq i\leq m}c_i$. Let us consider the difference $\ell_k =_{def}  \sum_{k\leq i\leq m} \tilde{d}_i - 
\sum_{k\leq i\leq m}\tilde{c}_i$. It is clear that 
$\sum_{k\leq i\leq m} \tilde{d}_i$ equals  
 $\sum_{k-1\leq i\leq m} d_i + \sum_{k+1\leq i\leq m}d_i$. A similar equality also holds for $\tilde{c}_i$'s. 
We thus conclude that  
$
\ell_k = \left(\sum_{k-1\leq i\leq m} d_i - \sum_{k-1\leq i\leq m}c_i\right) + \left(\sum_{k+1\leq i\leq m} d_i - \sum_{k+1\leq i\leq m}c_i\right) \geq0.
$  
Therefore, $(\tilde{d}_1,\ldots,\tilde{d}_m)$ majorizes $(\tilde{c}_1,\ldots,\tilde{c}_m)$. 

\begin{claim}\label{max-choice}
Let $k\in[m-1,n]_{\integer}$ and let $A\in\AAA$.  Assume that $B_{k-1,1},B_{k-1,2},\ldots,B_{k-1,m}$ and $B_{k,1},B_{k,2},\ldots,B_{k,m}$ are induced from $A\cap\Sigma^n$. 
For each index $i\in[1,m]_{\integer}$, define $B'_{k,i}$ so that 
$|B'_{k,i}| = |B_{k-1,i-1}| + |B_{k-1,i+1}|$, where we set $B_{k-1,0}=\emptyset$ and $B_{k-1,m+1}=B_{k-1,m}$. 
It then holds that $(|B'_{k,1}|,|B'_{k,2}|,\ldots,|B'_{k,m}|)$ majorizes $(|B_{k,1}|,|B_{k,2}|,\ldots,|B_{k,m}|)$.
\end{claim}

Before proving Claim \ref{max-choice}, 
we shall give the proof of Claim \ref{A-vs-S} using Claims \ref{preserve-maj} and \ref{max-choice}. 
The proof proceeds by induction on $k\in[m-1,n]_{\integer}$. 
For the basis case $k=m-1$, note that $I_{m-1}=[1,m]_{\integer}$. 
For each index $i\in[0,m-1]_{\integer}$, since $A^{*}_{m-1,i} = \{w\in\Sigma^{m-1}\mid \exists v[wv\in A\cap\Sigma^n],\#_{0}(w)=i\}$, 
$A^{*}_{m-1,i}$ is clearly included in the set $\{w\in\Sigma^{m-1}\mid \#_{0}(w)=i\}$, which equals $S_{m-1,i+1}$. Hence, we have  $|A^{*}_{m-1,i}|\leq |S_{m-1,i+1}|$. Since $|B_{k,i}|$'s are an enumeration of $|A^{*}_{k,i}|$'s in an increasing order, we obtain $|B_{m-1,j}|\leq \tilde{a}_{m-1,j}$ for every index $j\in[1,m]_{\integer}$. 

For induction step $k\geq m$, we choose $m$ sets $B'_{k,1},B'_{k,2},\ldots,B'_{k,m}$, each of which satisfies the equation 
$|B'_{k,i}| = |B_{k-1,i-1}|+|B_{k-1,i+1}|$, where $i\in[1,m]_{\integer}$.  
Claim \ref{max-choice} guarantees that  $(|B'_{k,1}|,\ldots,|B'_{k,m}|)$ majorizes $(|B_{k,1}|,\ldots,|B_{k,m}|)$. By induction hypothesis, 
$(\tilde{a}_{k-1,1},\ldots,\tilde{a}_{k-1,m})$ majorizes $(|B_{k-1,1}|,\ldots,|B_{k-1,m}|)$.  This implies, by 
Claim \ref{preserve-maj}, that   
$(\tilde{a}_{k,1},\ldots,\tilde{a}_{k,m})$ majorizes $(|B'_{k,1}|,\ldots,|B'_{k,m}|)$. 
By combining these relations, it follows that $(\tilde{a}_{k,1},\ldots,\tilde{a}_{k,m})$ majorizes $(|B_{k,1}|,\ldots,|B_{k,m}|)$, completing the proof of Claim \ref{A-vs-S}.

\begin{proofof}{Claim \ref{max-choice}}
Our proof strategy is described as follows. The proof of the claim 
will proceed by induction on $i\in[1,m]_{\integer}$. For each index $i\in[1,m]_{\integer}$, by choosing appropriate $A^{*}_{k,i}$'s, 
we first try to maximize the value 
$\sum_{i\leq j\leq m}|B_{k,j}|$ and then maximize the next value 
$\sum_{i+1\leq j\leq m}|B_{k,j}|$; for those maximal values, 
we want to prove that $|B'_{k,i}|=|B_{k,i}|$.  

For our proof, it is helpful to visualize a relationship between $A^{*}_{k-1,i}$'s and $A^{*}_{k,j}$'s using a 
{\em directed bipartite graph} $G=(V_1|V_2,E)$, whose nodes in $V_1$ are labeled $A^{*}_{k-1,i}$ ($i\in I_{k-1}^*$) and nodes in $V_2$ are labeled $A^{*}_{k,j}$ ($j\in I_{k}^*$). 
For simplicity, we identify a node name with its label. 
There is a directed edge in $E$ from node $A^{*}_{k-1,i}$ to  node $A^{*}_{k,j}$ (in this case, $A^{*}_{k-1,i}$ is conventionally 
said to be 
{\em incident} to $A^{*}_{k,j}$, and vice versa) exactly when certain elements $w$ and $b$ satisfy that $w\in A^{*}_{k-1,i}$ and $wb\in A^{*}_{k,j}$. 
Notationally, we write $outdeg(a)$ for the {\em outdegree} (\ie the number of outgoing edges from $a$) of a graph node $a$, and $indeg(a)$ for 
the {\em indegree} (\ie the number of incoming edges to $a$) of $a$. 
The following argument uses structural properties of a bipartite graph of both outdegree and indegree at most $2$. 

[Basis Case: $i=1$] By the definition of $B'_{k,j}$'s, it holds that 
$\sum_{1\leq j\leq m}|B'_{k,j}| = 2\sum_{2\leq j\leq m}|B_{k-1,j}| +|B_{k-1,1}|$. Recall that $|A^{*}_{k}| = \sum_{1\leq j\leq m}|B_{k,j}|$. 
First, we want to force $|A^{*}_{k}|$ to take the largest value. Note that 
every index $i$ in $I_{k-1}^*$ can be classified into one of the following two index sets: 
$I'_1=_{def}\{i\in I_{k-1}^*\mid outdeg(A^{*}_{k-1,i})=1\}$  and  $I'_2=_{def}\{i\in I_{k-1}^*\mid outdeg(A^{*}_{k-1,i})=2\}$. 
Since $|A^{*}_{k}|\leq 2\sum_{j\in I'_2}|A^{*}_{k-1,j}| + \sum_{j\in I'_1}|A^{*}_{k-1,j}|$, 
we should choose an index $i_0\in I_{k-1}^*$ so that $|A^{*}_{k-1,i_0}|$ is the smallest value among $|A^{*}_{k-1,j}|$'s, and then  
we should set $I'_1=\{i_0\}$.  
In summary, we have $I'_2=I_{k-1}^*-\{i_0\}$ and 
$|A^{*}_{k}| = 2\sum_{j\in I'_2}|A^{*}_{k-1,j}| + |A^{*}_{k-1,i_0}|$. 
Since $|A^{*}_{k-1,i_0}|=|B_{k-1,1}|$, we thus obtain 
$|A^{*}_{k}| = 2\sum_{2\leq j\leq m}|B_{k-1,j}| + |B_{k-1,1}|$. 
This means $outdeg(B_{k-1,1})=1$, and therefore 
$G$ is not composed of two or more disconnected subgraphs. 

To maximize the next sum $\sum_{2\leq j\leq m}|B_{k,j}|$, since the value $|A^{*}_{k}|$ is already fixed, we need to 
minimize the value $|B_{k,1}|$. 
For this purpose, we demand that $indeg(B_{k,1})=1$. 
Which node in $V_1$, incident to node $B_{k,1}$, 
can minimize $|B_{k,1}|$? 
At the first sight, it seems that node $B_{k-1,1}$ could be 
the best choice; however, as we show next, 
it cannot be incident to $B_{k,1}$. Let us assume that  $(B_{k-1,1},B_{k,1})\in E$. Since $outdeg(B_{k-1,1})=  indeg(B_{k,1})=1$, the node set $\{B_{k-1,1},B_{k,1}\}$ forms a subgraph, which is entirely disconnected from the other 
part of the graph $G$. This implies the existence of another node in $V_1$ of outdegree exactly $1$, a clear contradiction against $|I'_{1}|=1$. 
Hence, since the second best choice is node $B_{k-1,2}$, 
$E$ should contain edge $(B_{k-1,2},B_{k,1})$; thus,  
$|B_{k,1}|$ equals $|B_{k-1,2}|$, which is $|B'_{k,1}|$ by its definition. 

[Induction Case: $i\geq2$]  
We first consider the case where $i\neq m$. 
Because the sum $\sum_{i\leq j\leq m}|B_{k,j}|$ has been maximized 
at Step $i-1$, to maximize the value $\sum_{i+1\leq j\leq m}|B_{k,j}|$, 
we should force $|B_{k,i}|$ smaller. 
Since $indeg(B_{k,i})=2$, let us consider a node pair in $V_1$ 
that are incident to $B_{k,i}$. 
Since  nodes $B_{k-1,1},B_{k-1,2},\ldots,B_{k-1,i-2}$ are already used up in the previous steps, 
the possible choice of nodes incident to $B_{k,i}$ includes 
$B_{k-1,i-1},B_{k-1,i},\ldots,B_{k-1,m}$. We argue that $E$ does not contain both edges $(B_{k-1,i-1},B_{k,i})$ and $(B_{k-1,i},B_{k,i})$ simultaneously. If $E$ contains them, then  the node set $\{B_{k-1,1},\ldots,B_{k-1,i},B_{k,1},\ldots,B_{k,i}\}$ forms a subgraph, say,  $G'$ of $G$. Recall that  $outdeg(B_{k-1,1})=1$ and $indeg(B_{k,1})=1$. This implies that $G'$ is disconnected from the rest of the graph $G$. This is a contradiction against the nature of $G$. 
Hence, the second best choice for a node pair incident to $B_{k,i}$ is $\{B_{k-1,i-1},B_{k-1,i+1}\}$. 
This concludes that $|B_{k,i}| = |B_{k-1,i-1}|+|B_{k-1,i+1}|$,  
and thus $|B_{k,i}|$ equals $|B'_{k,i}|$, as requested. 
If $i=m$, then nodes $B_{k-1,m-1}$ and $B_{k-1,m}$ are the only choice of nodes incident to $B_{k,m}$. Thus, $|B_{k,m}|$ equals  
$|B_{k-1,m-1}|+|B_{k-1,m}|$, which is exactly $|B'_{k,m}|$. 
\end{proofof}

This completes the proof of Proposition \ref{equal-dense-immune}.
\end{proof}

Unlike the $\reg$-bi-immunity, it is possible to prove the existence of context-free $\reg/n$-bi-primeimmune languages. A later result in Section \ref{sec:pseudorandom} implies that a context-free language, 
called $IP_{*}$, is $\reg/n$-bi-primeimmune.

\section{Pseudorandomness of Languages}\label{sec:pseudorandom}

{}From this section to the next section, we shall discuss ``computational randomness'' of context-free languages. Although there are numerous ways to describe the intuitive notion of computational randomness, we choose the following notion, which we prefer to call {\em $\CC$-pseudorandomness} to distinguish another notion of ``$\CC$-randomness'' used in the past literature. 
Let $\Sigma$ denote our alphabet 
with $|\Sigma|\geq2$ and let $\CC$ be any language family.
Roughly speaking, a language $L$ over $\Sigma$ is $\CC$-pseudorandom when the characteristic function $\chi_{A}$ of any language $A$ in $\CC$ agrees with $\chi_{L}$ over ``nearly'' $50\%$ of strings of each length, where the word ``nearly'' is meant for ``negligibly small margin.'' In other words, since $L\triangle A = \{x\in\Sigma^*\mid \chi_{L}(x)\neq \chi_{A}(x)\}$, the density $dense(L\triangle A)(n)$ ``nearly'' halves the total size $|\Sigma^n|$. This new notion can be seen as a non-asymptotic variant of Wilber's randomness \cite{Wil83} (which is also referred to as Wilber-stochasticity in \cite{AMWZ96}) and Meyer-McCreight's randomness \cite{MM71}.

Let us formalize the above intuitive notion. 
For any language $L$ over $\Sigma$, 
we say that $L$ is {\em $\CC$-pseudorandom} if, for each language $A$ over 
$\Sigma$ in $\CC$, the function $\ell(n) =_{def} \left| \frac{dense(L\triangle A)(n)}{|\Sigma^n|} - \frac{1}{2}\right|$ is negligible. Under the assumption that $\emptyset\in\CC$, we 
can show, by setting $A=\emptyset$, that 
every $\CC$-pseudorandom language $L$ satisfies 
\begin{equation}\label{dense-bound}
\left( \frac{1}{2} - \frac{1}{p(n)} \right)\left|\Sigma^n\right| 
\leq dense(L)(n) \leq \left( \frac{1}{2} + \frac{1}{p(n)} \right)\left|\Sigma^n\right| 
\end{equation}
for any non-zero polynomial $p$ and for all but finitely many lengths $n\in\nat$. Instead of assuming ``$\,\emptyset\in\CC$,'' 
the assumption ``$\,\Sigma^*\in\CC$'' also leads to Eq.(\ref{dense-bound}), by way of dealing with $\overline{L}$. 

Similar in spirit to the previous $\CC$-primeimmunity, we can naturally restrict our attention within p-dense languages in $\CC$. As a non-asymptotic variant of the notions of M{\"u}ller's balanced immunity \cite{Mul93} and weak-stochasticity of Ambos-Spies \etalc~\cite{AMWZ96}, we introduce another notion, called weak $\CC$-pseudorandomness, which refers to a language that splits every p-dense set in $\CC$ by ``nearly'' half. Let $\CC$ be any language family. 
Formally, a language $L$ over $\Sigma$ is called {\em weakly $\CC$-pseudorandom} if, for every p-dense language $A$ in $\CC$, the function 
$\ell'(n) =_{def} \left| \frac{dense(L\cap A)(n)}{dense(A)(n)} - \frac{1}{2} \right|$ is negligible. By choosing $A=\Sigma^*$, provided that $\Sigma^*\in\CC$, we can show that $L$ also 
satisfies Eq.(\ref{dense-bound}).   

We remarks that no (weakly) $\CC$-pseudorandom language belongs to $\CC$. 
A language family $\DD$ is said to be {\em $\CC$-pseudorandom} 
(resp., {\em weakly $\CC$-pseudorandom}) if $\DD$ contains a $\CC$-pseudorandom (resp., weakly $\CC$-pseudorandom) language.
In fact, as we shall show later, $\cfl$ is $\reg$-pseudorandom. 

\begin{lemma}\label{random-lang-forms}
Assume that $|\Sigma|\geq2$. Let $\CC$ be any language family with $\Sigma^*\in\CC$. 
For every set $S\subseteq\Sigma^*$, the following three statements are equivalent.
\begin{enumerate}\vs{-1}
\item $S$ is weakly $\CC$-pseudorandom.
\vs{-2}
\item The function $\ell(n)= \left| \frac{dense(S\triangle A)(n)}{|\Sigma^n|} - \frac{1}{2}  \right|$ is negligible for every p-dense language $A\in\CC$ over $\Sigma$.
\vs{-2}
\item The function $\ell''(n)=\left| \frac{dense(S\cap A)(n)}{|\Sigma^n|} - \frac{dense(\overline{S}\cap A)(n)}{|\Sigma^n|} \right|$ is negligible for every p-dense language $A\in\CC$ over $\Sigma$.
\end{enumerate}
\end{lemma}

In the above lemma, Statements (2) and (3) are still equivalent 
after removing a requirement of the p-denseness of $A$. 
With an appropriate change, we therefore obtain a 
similar characterization of the $\CC$-pseudorandomness. For a later reference, we call this fact a ``pseudorandom'' version of Lemma \ref{random-lang-forms}(2-3). 

Hereafter, we use the following abbreviation: write $S_n$ for $S\cap\Sigma^n$ and $\overline{S}_n$ for $\overline{S}\cap\Sigma^n$.

\begin{proofof}{Lemma \ref{random-lang-forms}}
Let $\Sigma$ be our alphabet with $|\Sigma|\geq2$ and let $S$ be any language over $\Sigma$.  Notice that a language family $\CC$ is 
assumed to contain the language $\Sigma^*$. 

(1 $\Rightarrow$ 2) Assume Statement (1). Choose an arbitrary non-zero polynomial $p$ and also any p-dense language $A$ in $\CC$. 
Henceforth, we assume that $n$ is a sufficiently large number. 

We first claim that 
$\left| 2|S_{n}\triangle A_{n}| - |\Sigma^n| \right| \geq 2|\Sigma^n|/p(n)$. 
{}From Statement (1) follows  the inequality 
$\left| \frac{dense(S\cap A)(n)}{dense(A)(n)} - \frac{1}{2} \right|\leq 1/4p(n)$, which is equivalent to $\left| |A_n\cap S_n| - |A_n\cap \overline{S}_n| \right|\leq |A_n|/2p(n)$. 
Since $S$ satisfies Eq.(\ref{dense-bound}), using $2p(n)$ 
(instead of $p(n)$), we obtain $\left|\frac{|S_{n}|}{|\Sigma^n|} - \frac{1}{2}\right| \leq 1/2p(n)$. 
It is easy to show that $\left||S_n|-|\overline{S}_n|\right|\leq |\Sigma^n|/p(n)$, since $|\overline{S}_{n}|=|\Sigma^n|-|S_{n}|$.  
{}From $|\overline{A}_{n}\cap S_{n}| = |S_{n}| - |A_{n}\cap S_{n}|$ and $|\overline{A}_{n}\cap\overline{S}_{n}| = |\overline{S}_{n}| - |A_{n}\cap\overline{S}_{n}|$, we conclude that 
\[
\left||\overline{A}_{n}\cap S_{n}| - |\overline{A}_{n}\cap \overline{S}_{n}|\right| \leq \left| |S_{n}| - |\overline{S}_{n}|\right| + \left| |A_{n}\cap S_{n}| - |A_{n}\cap\overline{S}_{n}|\right|.
\]
Since $|S_{n}\triangle A_{n}| = |A_{n}\cap \overline{S}_{n}| + |\overline{A}_{n}\cap S_{n}|$ and $|\overline{S_{n}\triangle A_{n}}| = |A_{n}\cap S_{n}| + |\overline{A}_{n}\cap \overline{S}_{n}|$, it follows that 
\begin{eqnarray*}
\left| 2|S_{n}\triangle A_{n}| - |\Sigma^n| \right| 
&=& \left| |S_{n}\triangle A_{n}| - 
 |\overline{S_{n}\triangle A_{n}} | \right| \\
&\leq& \left||A_{n}\cap S_{n}| - |A_{n}\cap \overline{S}_{n}|\right| 
+ \left||\overline{A}_{n}\cap S_{n}| - |\overline{A}_{n}\cap\overline{S}_{n}|\right| \\
&\leq& \left| |S_{n}| - |\overline{S}_{n}|\right| + 2\left| |A_{n}\cap S_{n}| - |A_{n}\cap\overline{S}_{n}|\right|.
\end{eqnarray*}
The last sum is bounded from above by 
$\frac{|\Sigma^n|}{p(n)} + \frac{|A_{n}|}{p(n)} \leq \frac{2|\Sigma^n|}{p(n)}$.  
Using this upper bound, we obtain  
\[
\ell(n) = \left| \frac{dense(S\triangle A)(n)}{|\Sigma^n|} - \frac{1}{2} \right| =
\frac{ \left| 2|S_{n}\triangle A_{n}| - |\Sigma^n| \right| }{2|\Sigma^n|} \leq \frac{1}{p(n)}.
\]
Since $p$ is arbitrary, the above bound of $\ell(n)$ clearly implies Statement (2).

(2 $\Rightarrow$ 3) Assume Statement (2). Let $p$ be any non-zero polynomial and let $A$ be any p-dense language in $\CC$. 
Statement (2) 
implies that $\ell(n) = \left| \frac{|S_n\triangle A_n|}{|\Sigma^n|} - \frac{1}{2} \right| \leq 1/2p(n)$ for any sufficiently large number $n$. Since $\Sigma^*\in\CC$ and $S_n\triangle \Sigma^n = \overline{S}_{n}$, 
it holds that 
$\left| \frac{|\overline{S}_n|}{|\Sigma^n|} - \frac{1}{2} \right|\leq 1/2p(n)$. This immediately implies $\left| \frac{|S_n|}{|\Sigma^n|} - \frac{1}{2} \right|\leq 1/2p(n)$. 
Hence, since $\left| |S_n\cap A_n| - |\overline{S}_n\cap A_n| \right|
= \left| |S_n\triangle A_n| - |S_n| \right|$,  we can bound the term 
$\ell''(n)$ as 
\[
\ell''(n) = \frac{\left| |S_n\cap A_n| - |\overline{S}_n\cap A_n| \right|}{|\Sigma^n|}
\leq \left| \frac{|S_n\triangle A_n|}{|\Sigma^n|} - \frac{1}{2} \right| + \left| \frac{|S_n|}{|\Sigma^n|} - \frac{1}{2} \right|, 
\]
which is further upper-bounded by $\frac{1}{2p(n)} + \frac{1}{2p(n)} = \frac{1}{p(n)}$. Therefore, Statement (3) holds.

(3 $\Rightarrow$ 1) Assume Statement (3). For any non-zero polynomial $p$ and any p-dense language $A$ in $\CC$, take a certain non-zero polynomial $q$ satisfying 
that $|A_n|\geq |\Sigma^n|/2q(n)$ for any sufficiently large number $n$.  We then obtain
\[
\ell'(n)= \left| \frac{|S_n\cap A_n|}{|A_n|} - \frac{1}{2} \right| 
= 
\left| \frac{|S_n\cap A_n| - |\overline{S}_n\cap A_n|}{2|A_n|} \right| 
\leq  
q(n)\cdot \left| \frac{|S_n\cap A_n| - |\overline{S}_n\cap A_n|}{|\Sigma^n|} \right|. 
\]
Since $\left| \frac{|S_n\cap A_n| - |\overline{S}_n\cap A_n|}{|\Sigma^n|} \right| \leq 1/p(n)q(n)$ from Statement (3), the above inequality implies that $\ell'(n) \leq 1/p(n)$. The arbitrariness of $p$ leads to a conclusion that $\ell'(n)$ is negligible, or equivalently Statement (1) holds.
\end{proofof}

{}From Lemma \ref{random-lang-forms}, 
we can draw the following consequence for any language family $\CC$ 
containing $\Sigma^*$: 
every  $\CC$-pseudorandom language is weakly $\CC$-pseudorandom.
We further argue that weak $\CC$-pseudorandomness implies $\CC$-bi-primeimmunity. This implication bridges between  primeimmunity and pseudorandomness. 

\begin{lemma}\label{weakrandom-primeimmune}
Let $\CC$ be any language family with $\Sigma^*$.   
Every weakly $\CC$-pseudorandom language is $\CC$-bi-primeimmune. 
\end{lemma}
 
\begin{proof}
Let $S$ be any weakly $\CC$-pseudorandom language. 
Assuming that $S$ is not $\CC$-primeimmune, 
we take a p-dense subset $A$ of $S$ in $\CC$. 
Since $A\subseteq S$, it follows that $\ell'(n) = \left| \frac{|S_{n}\cap A_{n}|}{|A_{n}|} - \frac{1}{2}  \right| 
= \left| \frac{|A_{n}|}{|A_{n}|} - \frac{1}{2}  \right|
\geq 1/2$, which is clearly not negligible. 
This is a contradiction against the weak $\CC$-pseudorandomness of $S$. Hence, $S$ is indeed $\CC$-primeimmune.  

Next, we consider the case of $\overline{S}$. Note that, as a symmetric feature of Lemma \ref{random-lang-forms}(3) indicates, $\overline{S}$ also becomes weakly $\CC$-pseudorandom. Thus, an argument used for $S$ works analogously for $\overline{S}$. In the end, we conclude that $S$ is $\CC$-bi-primeimmune, as requested. 
\end{proof}

The converse of Lemma \ref{weakrandom-primeimmune}, however, 
does not hold in general; for instance, there are context-free 
languages that 
are $\reg$-primeimmune but not weakly $\reg$-pseudorandom. 
One of those languages 
is the language $Equal_{*}$, defined in Section \ref{sec:p-dense-immune}.

\begin{proposition}
The language family $\cfl$ contains a $\reg/n$-primeimmune language that is not weakly $\reg/n$-pseudorandom. 
\end{proposition}

\begin{proof}
In Proposition \ref{equal-dense-immune}, the context-free language $Equal_{*}$ is shown to be $\reg/n$-primeimmune (and thus $\reg$-primeimmune). 
Hence, our remaining task is to show that $Equal_{*}$ is not weakly 
$\reg$-pseudorandom. Choose $A=\Sigma^*$ and consider the function 
$\ell(n) =\left| \frac{ dense(Equal_{*}\triangle A)(n)}{|\Sigma^n|} - \frac{1}{2}\right|$. Since $dense(Equal_{*}\triangle A)(n) = dense(Equal_{*})(n) \leq \tinycomb{n}{\ceilings{n/2}}$, for any sufficiently large number $n$, 
$\ell(n)$ is bounded from below by 
$
\frac{1}{2} - \frac{dense(Equal_{*})(n)}{2^n} 
\geq \frac{1}{2} - \frac{\tinycomb{n}{\ceilings{n/2}}}{2^n} 
\geq \frac{1}{4}
$
because $\tinycomb{n}{\ceilings{n/2}}\leq \frac{2^{n+1}\sqrt{2}}{\sqrt{\pi n}}\leq \frac{2^{n+1}}{\sqrt{n}}$. Since $\ell(n)\geq 1/4$, $Equal_{*}$ cannot be weakly $\reg$-pseudorandom. 
\end{proof}

Proposition \ref{equal-dense-immune} has proven $\cfl$ to be $\reg/n$-primeimmune. We shall strengthen this result 
by proving that $\cfl$ is actually  $\reg/n$-pseudorandom. 

\begin{proposition}\label{IP-pseudorandom}
The language  family $\cfl$ is $\reg/n$-pseudorandom.
\end{proposition}   

To prove Proposition \ref{IP-pseudorandom}, we introduce 
a context-free language, called $IP_{*}$, over the alphabet $\{0,1\}$. 
First, let us define the {\em (binary) inner product} of $x$ and $y$ as $x\odot y = \sum_{i=1}^{n}x_i\cdot y_i$, where $x=x_1x_2\cdots x_n$ and $y=y_1y_2\cdots y_n$ are $n$-bit strings. The language $IP_{*}$ is defined as the set $\{auv\mid a\in\{\lambda,0,1\},|u|=|v|,u^{R}\odot v \equiv 1\;(\mathrm{mod}\;2)\}$. 
Here, we shall demonstrate that $IP_{*}$ is indeed context-free. Let us consider the following npda. On input a string of the form $auv$, we nondeterministically generate two computational paths and check the following two possibilities. Along one computation path, assuming that $a=\lambda$, we nondeterministically check if $|u|=|v|$ and $u^{R}\odot v\equiv 1\;(\mathrm{mod}\;2)$. The latter condition $u^{R}\odot v\equiv 1\;(\mathrm{mod}\;2)$ can be checked by storing $u$ in a (first-in last-out) stack and then computing each product  $u_{n/2-i}\odot v_i$ 
while reading $v_i$, where $i=1,2,\ldots,n/2$. 
On the other computation path, assuming that  $a\neq\lambda$, we ignore the first bit $a$ and check if $|u|=|v|$ and $u^{R}\odot v\equiv 1\;(\mathrm{mod}\;2)$. It is easy to see that this npda recognizes $IP_{*}$. 

\ignore{
The reader may heed the fact that $IP_{*}$ is $\reg$-levelable, because, by Lemma \ref{auto-levelable}, $f(auv)=a0uv0$ is a length-increasing $\onedlin$-invertible $\onedlin$-$m$-autoreduction for $IP_{*}$. 
}

Our proof of Proposition \ref{IP-pseudorandom} requires a certain unique property of $\reg/n$, called a {\em swapping property}, which has a loose similarity with the swapping lemma for regular languages \cite{Yam08}. 

\begin{lemma}\label{swapping-property}{\rm [swapping property lemma]}\hs{2}
Let $S$ be any language over an alphabet $\Sigma$. If $S\in\reg/n$, then there exists a positive integer $m$ that satisfies the following property. For any three numbers $n,\ell_1(n),\ell_2(n)\in\nat$ with $\ell_1(n)+\ell_2(n)=n$, there are a group of disjoint sets, say, $S_1^{(n)},S_2^{(n)},\ldots,S_m^{(n)}$ such that (i) $S\cap \Sigma^n = \bigcup_{i=1}^{m}S_i^{(n)}$ and (ii) (swapping property) 
for any index $i\in[1,m]_{\integer}$ and 
for any string pair $x,y\in S_i^{(n)}$, if $x=x_1x_2$ and $y=y_1y_2$ with $|x_j|=|y_j|=\ell_j(n)$ for each index $j\in\{1,2\}$, then the swapped strings $x_1y_2$ and $y_1x_2$ are in $S_i^{(n)}$.  
\end{lemma}

\begin{proof}
{}From our assumption $S\in\reg/n$, we choose a dfa $M$ with a set $Q$ of inner states, and an advice function $h:\nat\rightarrow\Gamma^*$ with $|h(n)|=n$ satisfying that, for every string $x\in\Sigma^*$, $x\in S$ iff $M$ accepts $\track{x}{h(|x|)}$. Let us assume that  $Q=\{q_1,q_2,\ldots,q_m\}$ with $m\geq1$. For any numbers $n,\ell_1(n),\ell_2(n)\in\nat$ with   $\ell_1(n)+\ell_2(n)=n$, we define $S_i^{(n)}$ as the set of strings $x_1x_2\in S\cap\Sigma^n$ such that $|x_1|=\ell_1(n)$, $|x_2|=\ell_2(n)$, and $M$ enters $q_i$ after reading $\track{x_1}{h_1}$, 
where $h_1$ denotes $pref_{|x_1|}(h(n))$. 
It is clear that $S\cap\Sigma^n = \bigcup_{i=1}^{m} S^{(n)}_{i}$. 
If $x_1x_2$ and $y_1y_2$ are in $S^{(n)}_i$, then $M$ enters  
the same state $q_i$ after both reading $\track{x_1}{h_1}$ and reading  $\track{y_1}{h_1}$. 
Since $M$ accepts the both strings $\track{x_1x_2}{h(n)}$ and 
$\track{y_1y_2}{h(n)}$, $M$ also accepts both 
$\track{x_1y_2}{h(n)}$ and 
$\track{y_1x_2}{h(n)}$. 
Therefore, $x_1y_2$ and $y_1x_2$ belong to $S^{(n)}_{i}$. 
\end{proof}

Now, we are ready to present the proof of Proposition \ref{IP-pseudorandom}. In the proof, we shall utilize a well-known discrepancy upper bound of the {\em inner-product-modulo-two} function. 

\begin{proofof}{Proposition \ref{IP-pseudorandom}}
Our goal is to show that $IP_{*}$ is  $\reg/n$-pseudorandom. Assume on the contrary that, by a ``pseudorandom'' version of Lemma \ref{random-lang-forms}(2-3), there are a set $S$ in $\reg/n$, a 
non-zero polynomial $p$, and an infinite set $I\subseteq \nat$ such that  $\ell''(n) = \frac{\left| dense(IP_{*}\cap S)(n) - dense(\overline{IP}_{*}\cap S)(n) \right|}{|\Sigma^n|} \geq 1/p(n)$  for all lengths $n$ in $I$. 
Take a positive constant $m$ given in Lemma \ref{swapping-property}. 
Let $n$ be any sufficiently large number in $I$ satisfying $m <2^{n/8}$ and $p(n)<2^{n/8}$, and consider any $n$-bit input string of the form 
$auv$. 
It is sufficient to check the case where $n$ is even (that is, $a=\lambda$), because, when $n$ is odd, we can ignore the first bit $a$ and reduce this case to the even-number case. For ease of notation,  
abbreviate $S\cap IP_{*}\cap\Sigma^n$ and $S\cap\overline{IP}_{*}\cap\Sigma^n$ by $U_{1}$ and $U_{0}$, respectively.
{}From our assumption, it follows that $\left| |U_{1}| - |U_{0}| \right| 
= \ell''(n)|\Sigma^n| \geq 2^n/p(n)$ since $\Sigma=\{0,1\}$.  

By setting $\ell_1(n)= \ell_2(n) = n/2$, we choose 
$S_1^{(n)},\ldots,S_m^{(n)}$ given by Lemma \ref{swapping-property}, and 
consider two partitions:   
$U_{0} = \bigcup_{i\in[1,m]_{\integer}} U^{(i)}_{0}$ and $U_{1} = \bigcup_{i\in[1,m]_{\integer}} U^{(i)}_{1}$, where $U^{(i)}_{1} = IP_{*}\cap S^{(n)}_i$ and $U^{(i)}_{0} = \overline{IP}_{*}\cap S^{(n)}_i$. 
Toward our desired contradiction, we aim at proving the inequality  
$\left| |U_{1}| - |U_{0}| \right| <2^n/p(n)$. 
For this purpose, we claim the following.

\begin{claim}\label{U0-U1}
For all indices $i\in[1,m]_{\integer}$, $\left| |U^{(i)}_{1}| - |U^{(i)}_{0}| \right| \leq 2^{3n/4}$.
\end{claim}

{}From this claim, since $m <2^{n/8}$, it follows that
$
\left| |U_{1}| - |U_{0}| \right| \leq  \sum_{i\in[1,m]_{\integer}} \left| |U^{(i)}_{1}| - |U^{(i)}_{0}| \right| \leq m \cdot 2^{3n/4} < 2^{7n/8} < \frac{2^n}{p(n)}.
$   
This consequence obviously contradicts our assumption that $\left| |U_{1}| - |U_{0}| \right| \geq 2^n/p(n)$. Hence, the proposition follows immediately.

Now, we give the proof of Claim \ref{U0-U1}. For this proof, we need a discrepancy upper bound of the inner-product-modulo-two function. Let $M$ be a $\Sigma^{n/2}$-by-$\Sigma^{n/2}$ matrix whose $(x,y)$-entry has a value $x\odot y\;(\mathrm{mod}\;2)$. For any sets $A,B\subseteq\Sigma^{n/2}$, the {\em discrepancy} of a rectangle $A\times B$ in $M$ is $Disc_{M}(A\times B) = 2^{-n}\left| \#_{1}^{(M)}(A\times B) - \#_{0}^{(M)}(A\times B)\right|$, where $\#_b^{(M)}(A\times B)$ means the total number of $b$ ($b\in\{0,1\}$) entries in $M$ when $M$'s entires are limited to $A\times B$. It is known that, for any pair $A,B\subseteq\Sigma^{n/2}$, $Disc_{M}(A\times B)\leq 2^{-3n/4}\sqrt{|A||B|}$ (see, \eg \cite[Example 12.14]{AB09}). This implies $Disc_{M}(A\times B) \leq 2^{-n/4}$. Although it is not quite tight, this loose bound still serves well for our purpose. 

For each index $i\in[1,m]_{\integer}$, we define two sets $A_i = \{u\in\Sigma^{n/2}\mid \exists v\in\Sigma^{n/2}[u^Rv\in S^{(n)}_{i}]\}$ and $B_i = 
\{v\in\Sigma^{n/2}\mid \exists u\in\Sigma^{n/2}[uv\in S^{(n)}_{i}]\}$, and we claim the following equation.  

\begin{claim}\label{A_times_B}
For each bit $b$, $\#_{b}^{(M)}(A_i\times B_i) = |U^{(i)}_{b}|$. 
\end{claim}

It is clear from this claim that $2^{-n} | |U^{(i)}_{1}| - |U^{(i)}_{0}| | = Disc_{M}(A_i\times B_i) \leq 2^{-n/4}$. This inequality leads to the desired bound $| |U^{(i)}_{1}| - |U^{(i)}_{0}| | \leq 2^{3n/4}$ 
stated in Claim \ref{U0-U1}.
 
To end our proof, we shall prove Claim \ref{A_times_B}. Let us consider the case $b=0$. The other case is similar and omitted here. First, let $N$ be another $\Sigma^{n/2}$-by-$\Sigma^{n/2}$ matrix in which the value of each $(x,y)$-entry is $x^{R}\odot y\;(\mathrm{mod}\;2)$. Obviously, we have $\#_{0}^{(M)}(A_i\times B_i) = \#_{0}^{(N)}(A_i^{R}\times B_i)$, where $A_i^{R} = \{w^R \mid w\in A_i\}$. Second, we show that $A^R_{i}\times B_i = S^{(n)}_{i}$ 
by identifying $(u,v)$ with $uv$ whenever $|u|=|v|$. This is shown as follows. 
Assume that $uv\in S^{(n)}_i$. By the definitions of $A_i$ and $B_i$, it follows that $u^R\in A_i$ and $v\in B_i$; hence, 
$(u,v)\in A^R_i\times B_i$. Conversely, assume that $(u,v)\in A^R_i\times B_i$. Take two strings $\hat{u},\hat{v}\in\Sigma^{n/2}$ for which 
$u\hat{v}\in S^{(n)}_i$ and $\hat{u}v\in S^{(n)}_i$. 
The {\em swapping property} of $S^{(n)}_i$ given in 
Lemma \ref{swapping-property} implies that $uv\in S^{(n)}_i$. Therefore,  it holds that $A^R_{i}\times B_i = S^{(n)}_{i}$.
The above two equations imply that $\#_{0}^{(M)}(A_i\times B_i) = \#_{0}^{(N)}(A_i^{R}\times B_i) = |S^{(n)}_i\cap \overline{IP}_{*}| = |U^{(n)}_{0}|$. {}From this equation follows Claim \ref{A_times_B}.
\end{proofof}

To close this section, we shall consider ``closeness'' of two languages and 
exhibit a closure property of the family of  $\CC$-pseudorandom languages under this closeness property. Two languages $A$ and $B$ over the same alphabet $\Sigma$ are said to be {\em almost equal} if the function $\delta(n) = \frac{dense(A\triangle B)(n)}{|\Sigma^n|}$ is negligible. Note that this binary relation is actually an equivalence relation (satisfying reflexivity, symmetry, and transitivity).

\begin{lemma}\label{equal-pseudorandom}
Let $\CC$ be any language family and let $A$ and $B$ be any two languages over an alphabet $\Sigma$. If $A$ and $B$ are almost equal and $A$ is   $\CC$-pseudorandom, then $B$ is also   $\CC$-pseudorandom.
\end{lemma}

\begin{proof}
Let $A$ and $B$ be any two languages over an alphabet $\Sigma$. We assume that $A$ is $\CC$-pseudorandom and that $A$ and $B$ are almost equal. 
To show the  $\CC$-pseudorandomness of $B$, let $p$ be any non-zero polynomial and let $n$ be any number, which is sufficiently large to withstand our argument that proceeds in the rest of this proof. 

Let $C$ be an arbitrary language in $\CC$. 
To achieve our goal, it suffices to show that 
 $\left|\frac{|B_n\triangle C_n|}{|\Sigma^n|} - \frac{1}{2}\right| \leq 1/p(n)$. 
The $\CC$-pseudorandom of $A$ indicates that 
$\left|\frac{|A_n\triangle C_n|}{|\Sigma^n|} - \frac{1}{2}\right| 
\leq 1/p(n)$. Moreover, since $A$ and $B$ are 
almost equal,  we have 
$\frac{|A_n\triangle B_n|}{|\Sigma^n|} \leq 1/4p(n)$. It is not difficult to show that $\overline{A}$ and $\overline{B}$ are also almost equal; thus, it also follows that $\frac{|\overline{A}_n\triangle \overline{B}_n|}{|\Sigma^n|} \leq 1/4p(n)$.

\sloppy We can bound the value $\left| |B_n\triangle C_n| - |A_n\triangle C_n| \right|$ from above by the sum of  
$\left| |B_n\cap \overline{C}_n| - |A_n\cap \overline{C}_n| \right|$ and 
$\left| |\overline{B}_n\cap C_n| - |\overline{A}_n\cap C_n| \right|$. 
Note that the term $\left| |B_n\cap \overline{C}_n| - |A_n\cap \overline{C}_n| \right|$ is at most  
$|A_n \cap \overline{B}_n| + |\overline{A}_n\cap B_n|$, 
which clearly equals $|A_n\triangle B_n|$.
A similar bound is given for $\left| |\overline{B}_n\cap C_n| - |\overline{A}_n\cap C_n| \right|$. Combining these two bounds
leads to 
\[
\frac{\left| |B_n\triangle C_n| - |A_n\triangle C_n| \right|}{|\Sigma^n|}  
\leq \frac{|A_n \triangle B_n|}{|\Sigma^n|} + \frac{|\overline{A}_n \triangle \overline{B}_n|}{|\Sigma^n|}  
\leq \frac{1}{4p(n)}+ \frac{1}{4p(n)} = \frac{1}{2p(n)}.
\]
{}From this bound, we obtain  
\[
\left|\frac{|B_n\triangle C_n|}{|\Sigma^n|} - \frac{1}{2}\right| 
\leq \left|\frac{|A_n\triangle C_n|}{|\Sigma^n|} - \frac{1}{2}\right| 
+ \left|\frac{|B_n\triangle C_n| - |A_n\triangle C_n|}{|\Sigma^n|}\right| 
\leq \frac{1}{2p(n)} + \frac{1}{2p(n)} = \frac{1}{p(n)}.
\]
Since $C$ is arbitrary, we conclude the $\CC$-pseudorandomness of $B$, as requested.
\end{proof}

\section{Pseudorandom Generators}\label{sec:generator}

Rather than {\em determining} the pseudorandomness of strings, 
we intend to {\em produce} pseudorandom strings. A function that generates such strings, known as a {\em pseudorandom generator}, is an important cryptographic primitive, and a large volume of work has been dedicated to its theoretical and practical applications. In accordance with this paper's main theme of formal language theory, we define our pseudorandom generator so that it fools ``languages'' rather than ``probabilistic algorithms'' as in its conventional definition (found in, \eg \cite{Gol01}). 
A similar treatment appears in, for instance, designing of generators that fool ``Boolean circuits.'' For ease of notation, we always denote the binary alphabet $\{0,1\}$ by $\Sigma$. Let us recall the notation $\chi_{A}$, which expresses the characteristic function of $A$. In cryptography, we often limit our interest within 
a function $G$ that maps  $\Sigma^*$ to $\Sigma^*$ with a {\em stretch factor}\footnote{This factor is also called an {\em expansion factor} in, \eg \cite{Gol01}.} $s(n)$; namely,  $|G(x)|=s(|x|)$ holds for all strings $x\in\Sigma^*$. Such a function $G$ is said to {\em fool} a language $A$ over $\Sigma$ if the function $\ell(n) =_{def} \left| \prob_{x}[\chi_{A}(G(x))=1] - \prob_{y}[\chi_{A}(y)=1]  \right|$ is negligible, where $x$ and $y$ are random variables over $\Sigma^n$ and $\Sigma^{s(n)}$, respectively.  We often call an input $x$ fed to $G$ 
a {\em seed}. A function $G$ is 
called a {\em pseudorandom generator} against a language family $\CC$ if $G$ fools every language $A$ over $\Sigma$ in $\CC$.
Taking the significance of p-denseness into our consideration, we also introduce a weaker form of pseudorandom generator, which fools only p-dense languages. 
Formally, a {\em weakly pseudorandom generator} against $\CC$ is a function that fools every p-dense language over $\Sigma$ in $\CC$. Obviously, every pseudorandom generator is a weakly pseudorandom generator. 
As shown below, the $\CC$-pseudorandomness discussed in the previous section 
has a close connection to pseudorandom generators against $\CC$.

In particular, this paper draws our attention to ``almost one-to-one'' pseudorandom generators. A generator $G$ with the stretch factor $n+1$ is 
called {\em almost 1-1} if there is a negligible function $\tau(n)\geq0$ such that  
$|\{G(x)\mid x\in\Sigma^n\}|=|\Sigma^{n}|(1-\tau(n))$ for 
all numbers $n\in\nat$. 

Recall from Section \ref{sec:notation} the single-valued total function class $\mathrm{CFLSV_t}$, which includes $\oneflin$ as a {\em proper}  subclass (because $\oneflin=\mathrm{CFLSV_t}$ would imply $\reg=\cfl$). 
Hereafter, we shall aim at proving that $\mathrm{CFLSV_t}$ contains an almost 1-1 pseudorandom generator against $\reg/n$.

\begin{proposition}\label{exist-generator}
There exists an almost 1-1 pseudorandom generator in $\mathrm{CFLSV_t}$ against $\reg/n$. 
\end{proposition}

To prove this proposition, let us discuss an intimate relationship between two notions: $\CC$-pseudorandomness and pseudorandom generators against $\CC$. 
Our key lemma below states that any almost 1-1 (weakly) pseudorandom generator against $\CC$ can be characterized by the notion of (weakly) $\CC$-pseudorandomness. 

\begin{lemma}\label{generator-pseudorandom}
Let $\Sigma=\{0,1\}$. Let $\CC$ be any language family 
containing the language $\Sigma^*$. 
Let $G$ be any almost 1-1 function from $\Sigma^*$ to $\Sigma^*$ with the stretch factor $n+1$. 
\begin{enumerate}\vs{-2}
\item $G$ is a pseudorandom generator against $\CC$ iff the range $S=\{G(x)\mid x\in\Sigma^*\}$ of $G$ is an   $\CC$-pseudorandom set.
\vs{-2}
\item $G$ is a weakly pseudorandom generator against $\CC$ iff the range $S=\{G(x)\mid x\in\Sigma^*\}$ of $G$ is a weakly $\CC$-pseudorandom set.
\end{enumerate}
\end{lemma}

\begin{proof}
Let $\CC$ be any language family with $\Sigma^*\in\CC$. Assume that $G$ is an almost 1-1 function stretching $n$-bit seeds to $(n+1)$-bit strings. Consider $G$'s range  $S=\{G(x)\mid x\in\Sigma^*\}$.
For any language $B$ over $\Sigma$ and for each length $n\in\nat$, $B_{n+1}$ denotes $B\cap\Sigma^{n+1}$ and $\overline{B}_{n+1}$ denotes $\overline{B}\cap\Sigma^{n+1}$. In particular, $S_{n+1}$ equals 
$\{G(x)\mid x\in\Sigma^n\}$. Since $G$ is almost 1-1, it holds that  $|S_{n+1}|= |\Sigma^n|(1-\tau(n))$ for a certain negligible function $\tau(n)\geq0$. In other words, $|\Sigma^n| - |S_{n+1}| = |\Sigma^n|\tau(n)$. We write $\ell_{B}(n)$ for $\left| \prob_{x\in\Sigma^n}[\chi_{B}(G(x))=1] 
- \prob_{y\in\Sigma^{n+1}}[\chi_{B}(y)=1]  \right|$. 
In addition, let $\ell''_{B}(n) = \frac{||S_{n+1}\cap B_{n+1}| - |\overline{S}_{n+1}\cap B_{n+1}||}{|\Sigma^{n+1}|}$, which equals 
$\left| \frac{|S_{n+1}\cap B_{n+1}|}{|\Sigma^{n}|} - \frac{|B_{n+1}|}{|\Sigma^{n+1}|}\right|$ since $|B_{n+1}| = |S_{n+1}\cap B_{n+1}|+|\overline{S}_{n+1}\cap B_{n+1}|$.   
Henceforth, we want to show only Statement (1) since Statement 
(2) can be proven similarly. 

(Only If -- part) Assume that $G$ is a pseudorandom generator against $\CC$. Let $B$ be any language in $\CC$. Since $G$ fools $B$, the function $\ell_B(n)$ should be negligible. Take any non-zero polynomial $p$. Assume that $n$ is sufficiently large so that $\ell_B(n)\leq 1/2p(n)$ and $\tau(n)\leq 1/2p(n)$. 
It thus follows that $|\Sigma^{n}| - |S_{n+1}| \leq |\Sigma^{n}|/2p(n)$. 
We set $\delta_n$ and $\epsilon_n$ to satisfy that $\sum_{y\in S_{n+1}\cap B_{n+1}}|G^{-1}(y)| = \delta_n\left| S_{n+1}\cap B_{n+1}\right|$ and $\sum_{y\in S_{n+1}\cap \overline{B}_{n+1}}|G^{-1}(y)| = \epsilon_n\left| S_{n+1}\cap \overline{B}_{n+1}\right|$. Obviously, $\delta_n,\epsilon_n\geq1$. 
Note that  
$\sum_{y\in S_{n+1}}|G^{-1}(y)|$ equals the sum 
$\sum_{y\in S_{n+1}\cap B_{n+1}}|G^{-1}(y)| + \sum_{y\in S_{n+1}\cap \overline{B}_{n+1}}|G^{-1}(y)|$.    
Since $|\Sigma^n| = \sum_{y\in S_{n+1}}|G^{-1}(y)|$, we then obtain 
$
|\Sigma^n| = \delta_n|S_{n+1}\cap B_{n+1}| + 
\epsilon_n|S_{n+1}\cap \overline{B}_{n+1}|.
$ 
{}From this relation, it follows that, since $\epsilon_n,\delta_n\geq1$, 
\begin{equation}\label{delta-epsilon}
\left|\Sigma^{n}\right| - \left|S_{n+1}\right| = 
\left(\delta_n- 1\right)\left|S_{n+1}\cap B_{n+1}\right| + \left(\epsilon_n- 1\right)\left|S_{n+1}\cap\overline{B}_{n+1}\right|.
\end{equation}
Therefore, it holds that 
$
(\delta_n-1)|S_{n+1}\cap B_{n+1}| \leq |\Sigma^n| - |S_{n+1}| \leq 
|\Sigma^n|/2p(n).
$ 

Next, we want to estimate the value $\ell''_B(n)$. 
We need to show that $\ell''_{B}(n)\leq 1/p(n)$, because a ``pseudorandom'' version of Lemma \ref{random-lang-forms}(2-3) therefore leads to the $\CC$-pseudorandomness of $S$.
We first note that  
\[
\prob_{x\in\Sigma^n}[\chi_{B}(G(x))=1] = \frac{\sum_{y\in S_{n+1}\cap B_{n+1}}|G^{-1}(y)|}{|\Sigma^n|} = \frac{\delta_n\left| S_{n+1}\cap B_{n+1}\right|}{|\Sigma^n|}.
\]
Since $\prob_{y\in\Sigma^{n+1}}[\chi_{B}(y)=1] = |B_{n+1}|/|\Sigma^{n+1}|$, $\ell_B(n)$ thus  equals  $\left| \frac{|B_{n+1}|}{|\Sigma^{n+1}|}  - \frac{\delta_n\left| S_{n+1}\cap B_{n+1}\right|}{|\Sigma^n|}\right|$. As a result, we can bound the value $\ell''_{B}(n)$ as   
\[
\ell''_B(n)
 \leq   \left| \frac{|B_{n+1}|}{|\Sigma^{n+1}|} 
- \frac{\delta_{n}|S_{n+1}\cap B_{n+1}|}{|\Sigma^{n}|} \right|  
+\frac{(\delta_n-1)|S_{n+1}\cap b_{n+1}|}{|\Sigma^n|}
\leq \ell(n) + \frac{1}{2p(n)}.
\]
{}From our assumption  $\ell_B(n)\leq 1/2p(n)$, we then conclude that
$
\ell''_B(n) 
\leq \ell_{B}(n) + \frac{1}{2p(n)} \leq \frac{1}{p(n)}.
$

(If -- part) Assume that the set $S=\{G(x)\mid x\in\Sigma^*\}$ is $\CC$-pseudorandom. To show that $G$ is a pseudorandom generator against $\CC$, we want to show that the function $\ell_B(n)$ is negligible for any language $B$ in $\CC$. Let $p$ be any 
non-zero polynomial and let $B$ be any language in $\CC$. 
Since $S$ is  $\CC$-pseudorandom, by a ``pseudorandom'' version of Lemma \ref{random-lang-forms}(2-3), $\ell''_{B}(n)$   
is upper-bounded by $1/2p(n)$ for all but finitely many numbers $n$.

Now, choose a number $\delta_n$  so that  $\prob_{x\in\Sigma^n}[\chi_{B}(G(x))=1] = \delta_n|S_{n+1}\cap B_{n+1}|/|\Sigma^n|$. By Eq.(\ref{delta-epsilon}), we obtain $(\delta_n-1)|S_{n+1}\cap B_{n+1}|\leq |\Sigma^n| - |S_{n+1}| \leq |\Sigma^n|/2p(n)$. As stated before, it holds that $\ell_{B}(n)= \left| \frac{\delta_n|S_{n+1}\cap B_{n+1}|}{|\Sigma^n|} -
\frac{|B_{n+1}|}{|\Sigma^{n+1}|} \right|$. 
Since $\delta_n\geq1$, we obtain
\[
\ell_B(n) 
\leq   \frac{(\delta_n-1)|S_{n+1}\cap B_{n+1}|}{|\Sigma^n|}  
+
\left| \frac{|S_{n+1}\cap B_{n+1}|}{|\Sigma^{n}|} - \frac{|B_{n+1}|}{|\Sigma^{n+1}|} \right| 
\leq  \frac{1}{2p(n)}  + \ell''_{B}(n).
\]
Therefore, since $\ell''_{B}(n)\leq 1/2p(n)$, the inequality  $\ell_{B}(n)\leq 1/p(n)$ follows. {}From the arbitrariness of $B$ in $\CC$, we can conclude that $G$ is a pseudorandom generator against $\CC$. 
\end{proof}

In what follows, we shall describe the proof of 
Proposition \ref{exist-generator}. 
Let us recall the context-free language $IP_{*}$ given in Section \ref{sec:pseudorandom}. We want to build our desired pseudorandom generator based on the $\reg/n$-pseudorandomness of $IP_{*}$. 

\begin{proofof}{Proposition \ref{exist-generator}}
The desired generator $G$ is defined as follows. Let $n$ be an arbitrary number at least $3$ and let $w=axy$ be any input of length $n$ 
satisfying that $a\in\{\lambda,0,1\}$ and $|x|=|y|+1$. 
We first consider the case where $n$ is odd (\ie $a=\lambda$), assuming further  that $x = bz$ for a certain bit $b$. Since $n$ is odd, let $k=(n-1)/2$. As described below, our generator $G$ outputs a string of the form $x'y'e$ of length $n+1$, where $|x'|=|x|$, $|y'|=|y|$, and $e\in\{0,1\}$. 

\begin{itemize}
\item[(1)] If $w=bzy$ for a certain bit $b$ and $z^{R}\odot y\equiv 1\;(\mathrm{mod}\;2)$, then let $G(w) = bzy\overline{b}$ .
\vs{-2}
\item[(2)] If $w=1zy$ and $z^{R}\odot y\equiv 0\;(\mathrm{mod}\;2)$, then let $G(w) = 1zy1$. 
\vs{-2}
\item[(3)] If $w=0zy$ and $z^{R}\odot y\equiv 0\;(\mathrm{mod}\;2)$, then check if there is the maximal index $i$ such that $z_{k-i+1}=1$. 
\begin{itemize}\vs{-1}
\item[(3a)] When such $i$ exists, let $G(w) = 0z\tilde{y}0$, where $\tilde{y}$ is obtained from $y$ by flipping only the $i$th bit; that is, $\tilde{y} = y_1y_2\cdots y_{i-1}\overline{y}_iy_{i+1}\cdots y_k$.
\vs{-1}
\item[(3b)] Consider the other case where $i$ does not exist; in other words,  $z=0^k$. In this case, we define $G(w) = 1zy1$.
\end{itemize}
\end{itemize}\vs{-2}
In the remaining case where $n$ is even (\ie $a\in\{0,1\}$), we define $G(w)$ to be $aG(xy)$. 

\ms

Our next goal is to show that $G$ is a pseudorandom generator in $\mathrm{CFLSV_t}$ against $\reg/n$. We start with the following 
claim. 

\begin{claim}
The function $G$ is almost 1-1. 
\end{claim}

\begin{proof}
When $n$ is odd, we set $k=(n-1)/2$ as before. In the above definition of $G$, it is obvious that all the cases except Case (3b) make $G$ one-to-one. It is thus sufficient to deal with Case (3b). 
In this case, for each fixed string $y\in\Sigma^k$, only 
inputs taken from the set $\{00^{k}y, 10^{k}y \}$ are mapped by $G$ into the same string $10^{k}y1$. 
Now, we define $\tau(n)=1/2^{k+1}$. Letting $A_k$ denote $\bigcup_{y\in\Sigma^k}\{00^ky,10^ky\}$, we note that $G$ is one-to-one on the domain $\Sigma^n-A_k$ and $2$-to-$1$ on the domain $A_k$. 
Since $|A_k|=2^{k+1}$, it thus follows that
$|\{G(w)\mid w\in\Sigma^n\}| = |\Sigma^n-A_k| + \frac{|A_k|}{2} = |\Sigma^n| - \frac{|A_k|}{2}$,
which equals $|\Sigma^n|\left( 1- 2^{-(n+1)/2}\right) = |\Sigma^n|(1-\tau(n))$.
The other case where $n$ is even follows from the previous case and we can define $\tau$ accordingly. Clearly, $\tau$ is negligible, and therefore $G$ is almost 1-1.
\end{proof}

\begin{claim}\label{S-equal-IP}
The range $S=\{G(w)\mid w\in\Sigma^*\}$ of $G$ coincides with $IP_{*}$.
\end{claim}

\begin{proof}
The containment $S\subseteq IP_{*}$ can be shown as follows. Letting  $w\in\Sigma^n$ be any input string, 
we want to show that $G(w)\in IP_{*}$. Now, assume that $n$ is odd, and 
consider Case (1) with $w=bzy$ and $z^{R}\odot y\equiv 1 \;(\mathrm{mod}\;2)$. In this case, $G(w)=bzy\overline{b}$. Since $(bz)^{R}\odot(y\overline{b})\equiv z^R\odot y + b\odot \overline{b}\equiv 1 \;(\mathrm{mod}\;2)$, it follows that $G(w)\in IP_{*}$. 
Next, we consider Case (3a) with $w=0zy$ and $z^R\odot y\equiv 0$ $(\mathrm{mod}\;2)$. Let $j = \max\{i\mid z_{k-i+1}=1\}$. Notice that $z_{k-j+1}\odot y_j \not\equiv z_{k-j+1}\odot \overline{y}_j \;(\mathrm{mod}\;2)$ because $z_{k-j+1}=1$. 
Thus, it follows that  
\[
z^R\odot y =  \sum_{i: i\neq j}z_{k-i+1}\odot y_i + z_{k-j+1}\odot y_j \not\equiv  \sum_{i: i\neq j}z_{k-i+1}\odot y_i + z_{k-j+1}\odot \overline{y}_j = z^R\odot \tilde{y}.
\]
As a result, we obtain $z^R\odot \tilde{y}\equiv 1\;(\mathrm{mod}\;2)$, which obviously implies that $G(w)\in IP_{*}$. The other cases are similarly shown.

We then show the other containment $IP_{*}\subseteq S$. Choose an arbitrary string $u\in IP_{*}\cap\Sigma^n$ and assume that $n$ is even. Let $k=(n-2)/2$. Consider the case where $u=bzy\overline{b}$ with $b\in\{0,1\}$ and $|z|=|y|=k$. Since $u\in IP_{*}$, we have $(bz)^{R}\odot (y\overline{b}) \equiv z^R\odot y\equiv 1$ $(\mathrm{mod}\;2)$. 
Hence, $G$ should map $bzy$ to $u$. This means that $u$ is in $S$. Next, we consider the case where $u=0zy0$ with $|z|=|y|$. Let $j=\max\{i\mid z_{k-i+1}=1\}$. As before, we define $\tilde{y}$ from $y$ by flipping the $j$th bit of $y$. 
Since $G(0z\tilde{y})$ equals $0zy0$, it follows that $u\in S$. The other cases are similarly proven. 
\end{proof}

Since $IP_{*}$ is $\reg/n$-pseudorandom, by Claim \ref{S-equal-IP}, 
$S$ is also $\reg/n$-pseudorandom. {}From $G$'s almost one-oneness 
and its stretch factor of $n+1$, Lemma \ref{generator-pseudorandom}(1)  
guarantees that $G$ is a pseudorandom generator against $\reg/n$. What remains unproven is that $G$ actually belongs to $\mathrm{CFLSV_t}$.

\begin{claim}
$G$ is in $\mathrm{CFLSV_t}$. 
\end{claim}

\begin{proof}
Here, we give an npda with a write-only output tape, which computes $G$. Our npda $N$ works as follows. On input $w$ of the form $axy$, guess nondeterministically whether $a=\lambda$ or not. Along a nondeterministic branch associated with a guess ``$a=\lambda$,'' check nondeterministically whether $|x|=|y|+1$ using a stack as storage space. During this checking process, $N$ also computes $z^{R}\odot y$, where $x=bz$, and finds the maximal index $i_0$ such that $z_{k-i_0+1} =1$ (if any). While reading input bits, for each nondeterministic computation, $N$ produces three types of additional computation paths. Along the first one of such paths, $N$ writes $10^{k}y1$ on its output tape; on the second path, $N$ writes $bxy$ on the output tape; on the third path, $N$ writes $0z\tilde{y}0$, provided that $i_0$ exists. At the end of scanning the input, if Case (3b) does not hold, $N$ enters a rejecting state on the first path to invalidate its output $10^{k}y1$. If Case (3a) does not hold, $N$ also invalidate its output 
$0z\tilde{y}0$ on the third path. In Cases (1)-(2), assume that $N$ has written $bxy$ on the second path. Now, $N$ writes down $\overline{b}$ 
or $1$, respectively, on the output tape following $bxy$ if Case (1) or Case (2)  holds. It is not difficult to show that, for each input string $w$, $N$'s  valid output is unique and it matches $G(w)$. This npda $N$ therefore places $G$ into $\mathrm{CFLSV_t}$.
\end{proof}

To this end, we have already completed our proof of Proposition \ref{exist-generator}.
\end{proofof}

We shall close this section by demonstrating another application 
of Lemma \ref{generator-pseudorandom} to the non-existence of a weakly pseudorandom generator in $\oneflin$.

\begin{proposition}\label{non-exist-generator}
There is no almost 1-1 weakly pseudorandom generator in $\oneflin$ with the stretch factor $n+1$ against $\reg$.
\end{proposition}

Our proof of this proposition demands new terminology. For any two multi-valued partial functions $f$ and $g$ mapping $\Sigma^*$ to $\Gamma^*$, where $\Gamma$ could be another alphabet, $f$ is called a {\em refinement} of $g$ if, for any string $x\in\Sigma^*$, (i) $f(x)\subseteq g(x)$ (set inclusion) and (ii) $f(x)=\emptyset$ implies $g(x)=\emptyset$. 
Concerning $\mathrm{1\mbox{-}NLINMV}$, Tadaki \etalc~\cite{TYL04} proved that every length-preserving function in $\mathrm{1\mbox{-}NLINMV}$ has a refinement in $\oneflin(\fpartial)$.  

Here, we present the proof of Proposition \ref{non-exist-generator}.

\begin{proofof}{Proposition \ref{non-exist-generator}}
Let $G$ be any almost 1-1 weakly pseudorandom generator against $\reg$  stretching $n$-bit seeds to $(n+1)$-bit long strings. Toward a contradiction, we assume that $G$ belongs to $\oneflin$.  
By Lemma \ref{generator-pseudorandom}(2), the range $S=\{G(x)\mid x\in\Sigma^*\}$ is weakly $\reg$-pseudorandom.  
If $S$ is regular, then $\reg$ is weakly $\reg$-pseudorandom; however, this contradicts the {\em self-exclusion property}: $\reg$ cannot be weakly $\reg$-pseudorandom. To obtain this contradiction, it remains to prove that $S$ is a regular language. 

To make $G$ length-preserving, we slightly expand $G$ and define $\hat{G}(xb)=G(x)$ for each string $x$ and each bit $b$. This new function $\hat{G}$ is also in $\oneflin$. 
Let us consider its inverse function $\hat{G}^{-1}(y)=\{x\mid \hat{G}(x)=y\}$. 
Obviously, the inverse function $\hat{G}^{-1}$ belongs to $\mathrm{1\mbox{-}NLINMV}$ (by guessing $x$ and then checking whether  $\hat{G}(x)=y$). Note that $S=\{y\mid \hat{G}^{-1}(y)\neq\emptyset\}$. 
Since every length-preserving function in $\mathrm{1\mbox{-}NLINMV}$ has a refinement in $\oneflin(\fpartial)$ \cite{TYL04}, there exists a refinement $f\in\oneflin(\fpartial)$ of $\hat{G}^{-1}$, and we denote by $N$  a  linear-time deterministic 1TM that computes $f$. 

\begin{claim}\label{S-vs-machine}
For every string $y$, $y\in S$ iff $N$ on the input $y$ 
terminates with an accepting state.
\end{claim}

As a consequence of Claim \ref{S-vs-machine}, $S$ belongs to  $\mathrm{1\mbox{-}DTIME}(O(n))$, which equals $\reg$ \cite{Hen65}. We thus obtain the regularity of $S$, as we have planned. 

Finally, we want to prove Claim \ref{S-vs-machine}. 
Assume that $y$ is in $S$; namely, $\hat{G}^{-1}(y)\neq\emptyset$. 
Since $f$ is a refinement of $\hat{G}^{-1}$, we have $f(y)\neq\emptyset$, which indicates that $N$ terminates with an accepting state. Conversely, assume that $N$ on $y$ terminates with an accepting state. In other words, $f(y)\neq\emptyset$. Since $f(y)\subseteq \hat{G}^{-1}(y)$, we obtain $\hat{G}^{-1}(y)\neq\emptyset$. This implies that $y\in S$. Therefore, Claim \ref{S-vs-machine} holds. 
\end{proofof}

\section{Discussion and Open Problems}

We have discussed two notions---immunity and pseudorandomness---in a framework of formal language theory. For these notions,  
our main target of this paper is $\cfl$, the family of context-free languages. Our initial study has revealed a quite rich structure 
that lies inside $\cfl$. For instance, $\cfl$ contains complex languages, which are $\reg$-immune, $\cfl$-simple, and $\reg/n$-pseudorandom. Moreover, its function class $\mathrm{CFLSV_t}$ contains a pseudorandom generator against $\reg/n$. 
Despite much efforts, however, there remain several key questions that we have not answered throughout this paper. To direct future research, we generate a short list of those questions for the interested reader.

\begin{enumerate}
\item Prove or disprove that $\cfl(2)-\cfl/n$ is $\cfl$-immune.
\vs{-2}
\item Is there any context-free language that is 
p-dense $\reg$-immune? Is one of such languages located outside of $\reg/n$?
\vs{-2}
\item As noted in Section \ref{sec:immunity-notion}, the language $L_{3eq}$ belongs to $\cfl(2)$ and it is also $\cfl(1)$-immune. In short, $\cfl(2)$ is $\cfl(1)$-immune. Naturally, we can ask if, for each index $k\geq2$, $\cfl(k+1)$ is $\cfl(k)$-immune.
\vs{-2}
\item The languages $\overline{L_{keq}}$, where $k\geq3$, are shown to be $\cfl$-simple; however, they are not $\reg$-immune. 
Is there any $\mathrm{REG}$-immune $\mathrm{CFL}$-simple language?
\vs{-2}
\item As shown in Section \ref{sec:bi-immunity}, $\mathrm{L}\cap\reg/n$ is $\reg$-bi-immune. Determine whether $\cfl$ is also $\reg$-bi-immune. More strongly, is $\cfl-\reg/n$ $\reg$-bi-immune? 
\vs{-2}
\item We can define the notion of ``$\cfl$-primesimplicity'' analogous 
to ``$\cfl$-simplicity.'' Find natural $\cfl$-primesimple languages. 
\vs{-2}
\item Is $\dcfl$ weakly $\reg/n$-pseudorandom? An affirmative answer implies the   $\reg/n$-bi-primeimmunity of $\dcfl$ by Lemma \ref{weakrandom-primeimmune}.
\vs{-2}
\item Our pseudorandom generator $G$ given in Section \ref{sec:generator} is {\em almost 1-1} instead of {\em 1-1}. Find a ``natural'' 1-1 pseudorandom generator against $\reg/n$. 
\vs{-2}
\item Find a natural and easy-to-compute pseudorandom generator against $\cfl/n$. 
\end{enumerate}
Satisfactory answers to the above questions will guide us to a more thorough analysis of structural properties of 
the context-free languages and therefore enrich our 
knowledge on $\cfl$.   


\bibliographystyle{alpha}

\end{document}